%% file: max-bisection.tex
\documentclass[11pt]{article}

\def\full{1}

\def\showauthornotes{1}
\def\showtableofcontents{1}
\def\showkeys{0}
\def\showdraftbox{0}
\def\showcolorlinks{1}
\def\usemicrotype{1}
\def\showfixme{0}

\input{macros}

\let\pref=\prettyref

\newcommand{\eat}[1]{}

\usepackage{pdfpages}

\ifnum\full=1
\else
\fi

\title{Approximating CSPs with Global Cardinality Constraints Using SDP Hierarchies}

\author{Prasad Raghavendra \and Ning Tan}

\date{}

\begin{document}

\maketitle
\draftbox
\thispagestyle{empty}


\begin{abstract}

	
	This work is concerned with approximating constraint
satisfaction problems (CSPs) with an
	additional global cardinality constraints.  For example, \maxcut is a boolean CSP where the input is a
	graph $G = (V,E)$ and the goal is to find a cut $S \cup \bar S
	= V$ that maximizes the number
	of crossing edges, $|E(S,\bar S)|$.  The \maxbisection problem is a variant of \maxcut
	with an additional global constraint that each side of the cut
	has exactly half the vertices, i.e., $|S| = |V|/2$.
	Several other natural optimization problems like \minbisection and approximating Graph Expansion can be
	formulated as CSPs with global constraints.

	In this work, we formulate a general approach towards
	approximating CSPs with global constraints using SDP
	hierarchies.  To demonstrate the approach we present the
	following results:

	\begin{itemize}
	\item Using the Lasserre hierarchy, we present an algorithm
		that runs in time $O(n^{poly(1/\epsilon)})$ that given
		an instance of \maxbisection with value $1-\epsilon$,
		finds a bisection with value $1-O(\sqrt{\epsilon})$.
		This approximation is near-optimal (up to constant
		factors in $O()$) under the Unique Games Conjecture.
    \item By a computer-assisted proof, we show that the same
algorithm also achieves a 0.85-approximation for \maxbisection,
improving on the previous bound of 0.70 (note that it is \uniquegames
hard to approximate better than a 0.878 factor). The same algorithm also yields a 0.92-approximation for \maxtwosat with cardinality constraints.
	\item   For every CSP with a global cardinality constraints, we present a generic conversion from integrality gap
		instances for the Lasserre hierarchy to a {\it dictatorship
		test} whose soundness is at most integrality gap.
		Dictatorship testing gadgets are central to hardness
		results for CSPs, and a generic conversion of the
		above nature lies at the core of the tight Unique Games
		based hardness result for CSPs. \cite{Raghavendra08}

	\end{itemize}
\end{abstract}
\clearpage

\ifnum\showtableofcontents=1
{
\tableofcontents
\thispagestyle{empty}
\clearpage
 }
\fi

\setcounter{page}{1}
\ifnum\full=0 \vspace{-8pt}\fi

\section{Introduction} \label{sec:intro}

Constraint Satisfaction Problems (CSP) are a class of fundamental
optimization problems that have been extensively studied in
approximation algorithms and hardness of approximation.  In  a
constraint satisfaction problem, the input consists of a set of
variables taking values over a fixed finite domain (say $\{0,1\}$) and
a set of {\it local} constraints on them.  The constraints are {\it
local} in that each of them depends on at most $k$ variables for some
fixed constant $k$.  The goal is to find an assignment to the
variables that satisfies the maximum number of constraints.

Over the last two decades, there has been much progress in understanding the approximability of
CSPs.  On the algorithmic front, semidefinite programming (SDP) has been
used with great success in approximating several well-known CSPs such
as \maxcut \cite{GoemansW95}, \maxtwosat \cite{CharikarMM07a} and
\maxthreesat \cite{KarloffZ97}.  More recently, these
algorithmic results have been unified and generalized to the entire class of constraint
satisfaction problems \cite{RaghavendraS09b}.  With the development of PCPs and long code based reductions, tight
hardness results matching the SDP based algorithms have been shown for
some CSPs such as Max-3-SAT \cite{Hastad01}.  In a surprising
development under the Unique Games Conjecture, semidefinite programming based algorithms have been shown
to be optimal for \maxcut \cite{KhotKMO07}, \maxtwosat
\cite{Austrin07a} and more generally every constraint satisfaction problem \cite{Raghavendra08}.

Unfortunately, neither SDP based algorithms nor the hardness
results extend satisfactorily to optimization problems with {\it
non-local} constraints.  Part of the reason is that the nice framework of SDP based approximation algorithms and matching
hardness results crucially rely on the {\it locality} of the constraints
involved.  Perhaps the simplest non-local constraint would be to
restrict the cardinality of the assignment, i.e., the number of ones
in the assignment.  Variants of CSPs with even a single cardinality constraint are not well-understood.
Optimization problems of this nature, namely constraint satisfaction problems with global cardinality constraints are the primary focus of this work.
Several important problems such as \maxbisection, \minbisection,
\smallsetexpansion can be formulated as CSPs with a single global cardinality
constraint.

As an illustrative example, let us consider the \maxbisection problem
which is also part of the focus of this work.  The \maxbisection problem is
a variant of the much well-studied \maxcut problem
\cite{GoemansW95,KhotKMO07}.  In the \maxcut
problem the goal is to partition the vertices of the input graph in to
two sets while maximizing the number of crossing edges.  The
\maxbisection problem includes an additional cardinality constraint
that both sides of the partition have exactly half the vertices of the
graph.
The seemingly mild cardinality constraint appears to change the nature
of the problem.  While \maxcut\ admits a factor $0.878$ approximation
algorithm~\cite{GoemansW95}, the best known
approximation factor for \maxbisection\ equals
$0.7027$~\cite{FeigeL06}, improving on previous bounds of
$0.6514$~\cite{FriezeJ97}, $0.699$~\cite{Ye01}, and
$0.7016$~\cite{HalperinZ02}. These algorithms proceed by rounding
the natural semidefinite programming relaxation analogous to the
Goemans-Williamson SDP for \maxcut.  Guruswami \etal
\cite{GuruswamiMRSZ11} showed that this natural SDP relaxation has a large integrality gap: the SDP
optimum could be $1$ whereas every bisection might only cut less than
$0.95$ fraction of the edges!  In particular, this implies that none
of these algorithms guarantee a solution with value close to $1$ even
if there exists a perfect bisection in the graph.  More recently,
using a combination of graph-decomposition, bruteforce enumeration and
SDP rounding, Guruswami \etal \cite{GuruswamiMRSZ11} obtained an algorithm that
outputs a $1-O(\epsilon^{1/3}\log(1/\epsilon))$ bisection on a graph that has a
bisection of value $1-\epsilon$.

A simple approximation preserving
reduction from \maxcut\ shows that \maxbisection\ is no easier to
approximate than \maxcut\ (the reduction is simply to take two
disjoint copies of the \maxcut\ instance). Therefore, the factor
$16/17$ NP-hardness \cite{Hastad01,TrevisanSSW00} and the factor $0.878$ Unique-Games
hardness for \maxcut~\cite{KhotKMO07} also applies to the
\maxbisection problem.  In fact, a stronger hardness result of factor $15/16$
was shown in ~\cite{HolmerinK04} assuming $\mathrm{NP} \not\subseteq
\bigcap_{\gamma > 0}
\mathrm{TIME}(2^{n^\gamma})$.
Yet, these hardness results for \maxbisection are far from matching the
best known approximation algorithm that only achieves a $0.702$ factor.

\paragraph{SDP Hierarchies}

Almost all known approximation algorithms for constraint satisfaction
problems are based on a fairly minimal SDP relaxation of the problem.
In fact, there exists a simple semidefinite program with linear number
of constraints (see \cite{Raghavendra08, RaghavendraS09b}) that yields the
best known approximation ratio for every CSP. This leaves open the possibility that stronger SDP relaxations such as
those obtained using the Lovasz-Schriver, Sherali-Adams and Lasserre
SDP hierarchies yield better approximations for CSPs.  Unfortunately,
there is evidence suggesting that the stronger SDP relaxations yield
no better approximation for CSPs than the simple semidefinite program
suggested in \cite{Raghavendra08,RaghavendraS09b}.  First, under the
Unique Games Conjecture, it is \NP-hard to approximate any CSP to a
factor better than that yielded by the simple semidefinite program
\cite{Raghavendra08}.  Moreover, a few recent works \cite{KhotS09,
Tulsiani09, RaghavendraS09c} have constructed
integrality gap instances for strong SDP relaxations of CSPs, obtained via Sherali-Adams and Lasserre
hierarchies.  For instance, the integrality gap instances in
\cite{KhotS09,RaghavendraS09c} demonstrate that up to $(\log\log n)^c$
rounds of the Sherali-Adams SDP hierarchy yields no better approximation
to \maxcut than the simple Goemans-Williamson semidefinite program
\cite{GoemansW95}.

The situation for CSPs with cardinality constraints promises
to be different.  For the \balancedseparator problem -- a CSP with a
global cardinality constraint, Arora \etal \cite{AroraRV04} obtained
an improved approximation of $\sqrt{\log n}$ by appealing to a
stronger SDP relaxation with triangle inequalities.  In case of
\maxbisection, one of the components of the algorithm of
\cite{GuruswamiMRSZ11} is a {\it brute-force search} -- a technique
that could quite possibly be carried out using SDP hierarchies.

Despite their promise, there are only a handful of applications of SDP hierarchies
in to approximation algorithms, most notably to approximating graph
expansion \cite{AroraRV04}, graph coloring and hypergraph independent
sets.  Moreover, there are few general techniques to round solutions to
SDP hierarchies, and analyze their integrality gap.

In an exciting development, fairly general techniques to round
solutions to SDP hierarchies (particularly the Lasserre hierarchy) has
emerged in recent works by Barak \etal \cite{BarakRS11} and Guruswami
and Sinop \cite{GuruswamiS11}.
Both these works (concurrently and independently) developed a fairly
general approach to round solutions to the Lasserre hierarchy using
an appropriate notion of local-global correlations in the SDP solution.
As an application of the technique, both the works obtain a
subexponential time algorithm for the Unique
Games problem using the Lasserre SDP hierarchy.  These works also
demonstrate several interesting applications of the technique.

Barak \etal \cite{BarakRS11} obtain an algorithm for arbitrary
$2$-CSPs with an approximation guarantee depending on the spectrum of the
input graph.  Specifically, the result implies a quasi-polynomial time approximation scheme for every $2$-CSP on
low threshold rank graphs, namely graphs with few large eigenvalues.

Guruswami and Sinop \cite{GuruswamiS11} obtain a general algorithm to
optimize quadratic integer programs with positive semidefinite forms
and global linear constraints.  Several interesting
problems including {\it $2$-CSPs with  global cardinality constraints} such
as \maxbisection, \minbisection and \balancedseparator fall in to
the framework of \cite{GuruswamiS11}.  However, the approximation
guarantee of their algorithm depends on the spectrum of the input
graph, and is therefore effective only on the special class of
low threshold rank graphs.

\ifnum\full=0 \vspace{-8pt}\fi
\subsection{Our Results}
\ifnum\full=0 \vspace{-8pt}\fi
In this paper, we develop a general approach to approximate CSPs with
global cardinality constraints using the Lasserre SDP hierarchy.

We illustrate the approach with an improved approximation algorithm
for the \maxbisection and balanced \maxtwosat problems.  For the
\maxbisection problem, we show the following result.

\begin{theorem} \label{thm:max-bisection}
For every $\delta > 0$, there exists an algorithm for \maxbisection
that runs in time $O(n^{\poly(1/\delta)})$ and obtains the following
approximation guarantees,
\begin{itemize}
	\item The output bisection has value at least $0.85 - \delta$ times the
		optimal max bisection.
	\item For every $\epsilon > 0$, given an instance $G$ with a
		bisection of value $1-\epsilon$, the algorithm
		outputs a bisection of value at least $1 - O(\sqrt{\epsilon}) -
		\delta$.
\end{itemize}
\end{theorem}

Note that the approximation guarantee of $1- O(\sqrt{\epsilon})$ on
instances with $1-\epsilon$ is nearly optimal (up to constant factors
in the $O()$)
under the Unique Games Conjecture.  This follows from the
corresponding hardness of \maxcut and the reduction from \maxcut to
\maxbisection.

Our approach is robust in that it also yields similar
approximation guarantees to the more general $\alpha$-\maxcut problem
where the goal is to find a cut with exactly $\alpha$-fraction of
vertices on one side of the cut.  More generally, the algorithm also
generalizes to a weighted version of \maxbisection, where the vertices
have weights and the cut has approximately half the weight on each
side.  \footnote{Note that in the weighted case, finding any exact bisection is at least as
hard as subset-sum problem.}

The same algorithm also yields an approximation to the complementary
problem of \minbisection.  Formally, we obtain the following
approximation algorithm for \minbisection and
$\alpha$-\balancedseparator.

\begin{theorem} \label{thm:minbisection}
For every $\delta > 0$, there exists an algorithm running in time
$O(n^{O(\poly(1/\delta))})$, which given a graph with a bisection ($\alpha$-balanced separator) cutting
$\epsilon$-fraction of the edges, finds a bisection ($\alpha$-balanced separator) cutting at most
$O(\sqrt{\epsilon})+\delta$-fraction of edges.
\end{theorem}

Towards showing a matching hardness results for CSPs with cardinality
constraints, we construct a {\it dictatorship test} for these
problems.  Dictatorship testing gadgets lie at the heart of all
optimal hardness of approximation results for CSPs (both
$\NP$-hardness and unique games based hardness results).  In fact,
using techniques from the work of Khot \etal \cite{KhotKMO07}, any
dictatorship test for a CSP yields a corresponding unique games based
hardness result.  More generally, a large fraction of hardness of
approximation results (not necessarily CSPs) have an underlying
dictatorship testing gadget.

Building on earlier works, Raghavendra \cite{Raghavendra08} exhibited
a generic reduction that starts with an arbitrary integrality gap
instance for certain SDP relaxation of a CSP to a dictatorship test for the same CSP.
In turn, this implied optimal hardness results matching the
integrality gap of the SDP under the unique games conjecture.
Using techniques from \cite{Raghavendra08}, we exhibit a generic
reduction from integrality gap instances to the Lasserre SDP
relaxation of a CSP with cardinality constraints, to a dictatorship
test for the same.  While the reduction applies in general for every
CSP with cardinality constraints, for the sake of exposition, we
present the special case of \maxbisection.  For \maxbisection, we show
the following.
\begin{theorem}(Informal Statement)
For every $\epsilon,\delta > 0$, given an integrality gap instance for
$\poly(1/\epsilon)$-round Lasserre SDP for \maxbisection, with SDP
value $c$ and optimum integral value $s$, there exists a dictatorship
test for \maxbisection with completeness $c - O(\epsilon+\delta)$ and
soundness $s + O(\epsilon+\delta)$.
\end{theorem}
\ifnum\full=0
The formal statement of the result and its proof is deferred to the
full version of the paper.
\fi
\ifnum\full=1
The formal statement of the result  and its proof is presented in
\pref{sec:gaptodict}.
\fi
Unfortunately, this dictatorship test does not yet translate in to a
corresponding hardness result for \maxbisection.
First, observe that the framework of Khot \etal \cite{KhotKMO07} to
show unique games based hardness results does not apply to
\maxbisection due to the global constraint on the instance.
This is the same reason why the unique games conjecture is not known to imply
hardness results for \balancedseparator.   The reason being that the hard instances of these problems are
required to have certain global structure (such as expansion in case of
\balancedseparator).  In case of \maxbisection, a hard instance must
not decompose in to sets of small size ($\epsilon n$ vertices), else
the global balance condition can be easily satisfied by appropriately
flipping the cut in each set independently.  Gadget reductions from a unique games instance
preserve the global properties of the unique games instance such as lack of expansion. Therefore, showing
hardness for \balancedseparator or \maxbisection problems require a
stronger assumption such as unique games with expansion or the Small
Set expansion hypothesis \cite{RaghavendraS10}.

\ifnum\full=0 \vspace{-8pt}\fi

\section{Overview of Techniques} \label{sec:overview}

In this section, we outline the our approach to approixmating the
\maxbisection problem.  The techniques are fairly general and can be applied to other CSPs with global cardinality constraints.

\paragraph{Global Correlation}
For the sake of exposition, let us recall the Goemans and Williamson algorithm
for \maxcut.
Given a graph $G = (V,E)$, the Goemans-Williamson SDP relaxation for
\maxcut assigns a unit vector $v_i$ for every vertex $i \in V$, so as to maximize the average squared length $E_{i,j \in E} \norm{v_i - v_j}^2$ of the edges.  Formally, the SDP relaxation is given by,
$$ \textrm{maximize }    \E_{i,j \in E} \norm{v_i - v_j}^2  \ \textrm{ subject to } \norm{v}_i^2 = 1 ~\forall i \in V$$
The rounding scheme picks a random halfspace passing through the
origin and outputs the partition of the vertices induced by the
halfspace.  The value of the cut returned is guaranteed to be within a
$0.878$-factor of the SDP value.

The same algorithm would be an approximation for \maxbisection if the
cut returned by the algorithm was near-balanced, i.e., $|S| \approx
|V|/2$.  Indeed, the expected number of vertices on either side of the
partition is $|V|/2$, since each vertex $i \in V$ falls on a given
side of a random halfspace with probability $\frac{1}{2}$.

If the balance of the partition returned is concentrated around its
expectation then the Goemans and Williamson algorithm would yield a
$0.878$-approximation for \maxbisection.   However, the balance of the
partition need not be concentrated, simply because the values taken by
vertices could be highly correlated with each other!

\paragraph{SDP Relaxation}
To exploit the correlations between the vertices we use a $k$-round Lasserre SDP
\cite{Lasserre01} of \maxbisection for a sufficiently large constant $k$.  On a high
level, the solutions to a Lasserre's SDP hierarchy are vectors
that {\it locally behave} like a distribution over integral solutions.
The $k$-round Lasserre SDP has the following properties similar to
a true distribution over integral solutions.
\begin{itemize}
\item {\it Marginal Distributions} For any subset $S$ of vertices with $|S|\leq k$, the SDP will
yield a distribution $\mu_S$ on partial assignments to the vertices
($\{-1,1\}^S$).  The marginals of $\mu_S$, $\mu_T$ for a pair of
subsets $S$ and $T$ are consistent on their intersection $S \cap T$.
\item {\it Conditioning} Analogous to a true distribution over
integral solutions, for any subset $S \subseteq V$ with $|S| \leq k$
and a partial assignment $\alpha \in \{-1,1\}^S$, the SDP solution can
be conditioned on the event that $S$ is assigned $\alpha$.
\end{itemize}
A detailed description of the Lasserre's SDP hierarchy for \maxbisection and other CSPs will be given in \pref{sec:prelim}.

\paragraph{Measuring Correlations}

  In this work, we will use mutual
information as a measure of correlation between two random variables.
We refer the reader to \pref{sec:prelim} for the definitions of
Shannon entropy and mutual information.  The correlation between vertices $i$ and $j$ is given by
$$I_{\mu_{i,j}}(X_i; X_j) = H(X_i) - H(X_i|X_j) \mcom$$
where the random variables $X_i,X_j$ are sampled using the local distribution
$\mu_{i,j}$ associated with the Lasserre SDP.
An SDP solution will be termed {\it $\alpha$-independent} if the average
mutual information between random pairs of vertices is at most
$\alpha$, i.e., $\E_{i,j \in V}[I(X_i;X_j)] \leq \alpha$.

For most natural rounding schemes such as the halfspace-rounding, the variance
of the balance of the cut returned is
directly related to the average correlation between random pairs
of vertices in the graph.  In other words, if the rounding scheme is
applied to an $\alpha$-independent SDP solution then the variance of
the balance of the cut is at most $poly(\alpha)$.

\paragraph{Obtaining Uncorrelated SDP Solutions}


Intuitively, if it is the case that globally all the vertices are
highly correlated, then conditioning on the value of a vertex should
reveal information about the remaining vertices, therefore reducing
the total entropy of all the vertices.

Formally, let us suppose the $k$-round Lasserre SDP solution is not
$\alpha$-independent, i.e.,  $\E_{i,j \in V}[I(X_i;X_j)] > \alpha$.
Let us pick a vertex $i \in V$ at random, sample its value $b \in
\{-1,1\}$ and condition the SDP
solution to the event $X_i = b$.  This conditioning reduces the
average entropy of the vertices ($\E_{j \in V}[H(X_j)]$) by at least
$\alpha$ in expectation.  If the conditioned SDP solution is
$\alpha$-independent we are done, else we repeat the process.

The intital average entropy $\E_{j \in V}[H(X_j)]$ is at most $1$, and
the quantity always remains non-negative.  Therefore, within
$\frac{1}{\alpha}$ conditionings, the SDP solution will be
$\alpha$-independent.  Starting with a $k$-round Lasserre SDP
solution, this process produces a $k-t$ round $\alpha$-independent Lasserre
SDP solution for some $t > \frac{1}{\alpha}$.

\paragraph{Rounding Uncorrelated SDP Solutions}
Given an $\alpha$-independent SDP solution, for many natural rounding
schemes the balance of the output cut is concentrated around its
expectation.  Hence it suffices to construct rounding schemes that
output a balanced cut in expectation.  We exhibit a simple rounding
scheme that preserves the bias of each vertex individually, thereby
preserving the global balance property.   
The details of the rounding algorithm will be described in \pref{sec:rounding}.

\ifnum\full=0 \vspace{-8pt}\fi
\section{Preliminaries} \label{sec:prelim}

\ifnum\full=1
\paragraph{Constraint Satisfaction Problem with Global Cardinality Constraints}

In this section we formally define CSPs with global constraints.

\begin{definition} [Constraint Satisfaction Problems with Global Cardinality Constraints]
A constraint satisfaction problem with global cardinality constraints is specified by $\Lambda=([q],\mathbb{P},k,c)$ where $[q]=\{0,\ldots,q-1\}$ is a finite domain, $\mathbb{P}=\{P:[q]^{t} \mapsto [0,1]|t\leq k$\} is a set of payoff functions. The maximum number of inputs to a payoff function is denoted by $k$. The map $c: [q]\mapsto [0,1]$ is the cardinality function which satisfies $\sum_{i} c_i=1$. For any $0\leq i\leq q-1$, the solution should contain $c_i$ fraction of the variables with value $i$.

\end{definition}
\begin{remark}
Although some problems (\eg \balancedseparator) do not fix the cardinalities to be some specific quantities, they can be easily reduced to the above case.
\end{remark}

\begin{definition}
An instance $\Phi$ of constraint satisfaction problems with global
cardinality constraints $\Lambda=([q],\mathbb{P},k,c)$ is given by $\Phi=(\mathcal{V},\mathbb{P}_{\mathcal{V}},W)$ where
\begin{itemize}
\item $\mathcal{V}=\{x_1,\ldots,x_n\}$: variables taking values over $[q]$
\item $\mathbb{P}_{\mathcal{V}}$ consists of the payoffs applied to subsets $S$ of size at most $k$
\item Nonnegative weights $W=\{w_S\}$ satisfying $\sum_{|S|\leq k}w_S=1$. Thus we may interpret $W$ as a probability distribution on the subsets. By $S\sim W$, we denote a set $S$ chosen according to the probability distribution $W$
\item An assignment should satisfy that the number of variables with value $i$ is $c_i n$ (we may assume this is an integer).
\end{itemize}
\end{definition}
\fi
Here we give a few examples of CSPs with global cardinality
constraints.
\ifnum\full=0
We defer the formal definition of a CSP with global
cardinality constraint to the full version.
\fi
\begin{definition} [Max(Min) Bisection]
Given a (weighted) graph $G=(V,E)$ with $|V|$ even, the goal is to partition the vertices into two equal pieces such that the number (total weights) of edges that cross the cut is maximized (minimized).
\end{definition}
More generally, in an $\alpha$-\maxcut problem, the goal is to find a
partition having $\alpha n$ vertices on one side, while cutting the
maximum number of edges.  Furthermore, one could allow weights on the
vertices of the graph, and look for cuts with exactly
$\alpha$-fraction of the weight on one side.  Most of our techniques
generalize to this setting.

Throughout this work, we will have a weighted graph $G$ with weights
$W$ on the vertices.  The weights on the vertices are assumed to form
a probability distribution.  Hence the notation $i \sim W$ refers to
a random vertex sampled from the distribution $W$.

\begin{definition} [Edge Expansion]
Given a graph (w.l.o.g, we may assume it is a unweighted regular graph) $G=(V,E)$, and $\delta\in(0,1/2)$, the goal is to find a set $S\subseteq V$ such that $|S|=\delta |V|$ and the edge expansion of $S$: $\Phi(S)=\frac{E(S,\bar{S})}{d|S|}$ is minimized.
\end{definition}

\paragraph{Information Theoretic Notions}
\begin{definition} \label{def:entropy}
Let $X$ be a random variable taking values over $[q]$. The \emph{entropy} of $X$ is defined as
$$
H(X)\defeq -\sum_{i\in[q]}\Pr(X=i)\log\Pr(X=i)
$$
\end{definition}
\begin{definition}\label{def:mutualinfo}
Let $X$ and $Y$ be two jointly distributed variables taking values over $[q]$. The \emph{mutual information} of $X$ and $Y$ is defined as
$$
I(X;Y)\defeq\sum_{i,j\in [q]}\Pr(X=i,Y=j)\log \frac{\Pr(X=i,Y=j)}{\Pr(X=i)\Pr(Y=j)}
$$
\end{definition}
\begin{definition}\label{def:conditionalentropy}
Let $X$ and $Y$ be two jointly distributed variables taking values over $[q]$. The \emph{conditional entropy} of $X$ conditioned on $Y$ is defined as
$$
H(X|Y)=\E_{i\in [q]}[H(X|Y=i)]
$$
\end{definition}
\ifnum\full=1
We also give two well-known theorems in information theory below.
\begin{theorem} \label{thm:entropy-info}
Let $X$ and $Y$ be two jointly distributed variables taking value on $[q]$, then
$$
I(X;Y)=H(X)-H(X|Y)
$$
\end{theorem}
\begin{theorem} \label{thm:dataprocessing} (Data Processing
Inequality)
Let $X,Y,Z,W$ be random variables such that $H(X|W)=0$ and $H(Y|Z)=0$, \ie $X$ is fully determined by $W$ and $Y$ is fully determined by $Z$, then
$$
I(X;Y)\leq I(W;Z)
$$
\end{theorem}
\fi

\paragraph{Lasserre SDP Hierarchy for Globally Constrained CSPs} \label{subsec:lasserresdp}

\ifnum\full=1
Let $\Lambda=([q],\mathbb{P},k,c)$ be a CSP with global constraints
and $\Phi=(\mathcal{V},\mathbb{P}_{\mathcal{V}},W)$ be an instance of
$\Lambda$ on variables $X=\{x_1,...,x_n\}$. A solution to the
$k$-round Lasserre SDP consists of vectors $v_{S,\alpha}$ for all
vertex sets $S\subseteq V$ with $|S|\leq k$ and local assignments
$\alpha\in [q]^{S}$. Also for each subset $S\subseteq V$ with $|S|\leq
k$, there is a distribution $\mu_S$ on $[q]^{S}$. For two subsets
$S,T$ such that $|S|,|T|\leq k$, we require that the corresponding
distributions $\mu_S$ and $\mu_T$ are consistant when restricted to $S\cap T$. A Lasserre solution is feasible if for any $|S\cup T|\leq k$, $\alpha\in [q]^S$, $\beta\in [q]^T$, we have
$$
\langle v_{S,\alpha},v_{T,\beta}\rangle =\mathbb{P}_{\mu_{S\cup T}}\{X_S=\alpha, X_T=\beta\}
$$
The SDP also has a vector $I$ that denotes the constant $1$.  The global cardinality constraints can be written in terms of the
marginals of each variable. Specifically, for every $S$ with $|S|\leq k-1$ and $\alpha\in [q]^{S}$, we have
$$
\mathbb{E}_j \mathbb{P}_{\mu_{S\cup \{x_j\}}} (x_j=i | X_S=\alpha)=c_i
$$
The objective of the SDP is to maximize
$$
\mathbb{E}_{S\in W}\left(\sum_{\beta\in[q]^{S}}
P_S(\beta(S))\mathbb{P}_{\mu_S}(S,\beta)\right)
$$

\fi
\ifnum\full=0
Consider a CSP with global constraint over an alphabet $[q] =
\{1,2,\ldots,q-1\}$.  For the purpose of this extended abstract, $[q]
= \{0,1\}$.
A $k$-rounds Lasserre solution consists of vectors $v_{S,\alpha}$ for
all vertex sets $S\subseteq V$ with $|S|\leq k$ and local assignments
$\alpha\in [q]^{S}$. Also for each subset $S\subseteq V$ with $|S|\leq k$, there is a distribution $\mu_S$ on $[q]^{S}$. For two subsets $S,T$ such that $|S|,|T|\leq k$, we require that $\mu_S$ and $\mu_T$ are consistant when restricted on $S\cap T$. A Lasserre solution is feasible if for any $|S\cup T|\leq k$, $\alpha\in [q]^S$, $\beta\in [q]^T$, we have
$$
\langle v_{S,\alpha},v_{T,\beta}\rangle =\mathbb{P}_{\mu_{S\cup T}}\{X_S=\alpha, X_T=\beta\}
$$
Also, the marginal distribution of each variable should match the cardinality constraints. That is, for any $S$ with $|S|\leq k-1$ and $\alpha\in [q]^{S}$, we have
$$
\mathbb{E}_j \mathbb{P}_{\mu_{S\cup \{x_j\}}} (x_j=i | X_S=\alpha)=c_i
$$
The objective of the SDP is to maximize
$$
\mathbb{E}_{S\in W}w_{S}(\sum_{\beta\in[q]^{S}} P_S(\beta(S))\mathbb{P}_{\mu_S}(S,\beta))
$$
\fi

While the complete description of the Lasserre SDP hierarchy is
somewhat complicated, there are few properties of the hierarchy that
we need. The most
important property is the existence of consistent local
marginal distributions $\{\mu_{S}\}_{S \sse V,|S| \leq k}$ whose first two
moments match the inner products of the vectors. We
stress that even though the local distributions are consistent, there
might not exist a global distribution that agrees with all of them.  The second property of the $k$-round Lasserre SDP
solution is that although the variables are not jointly distributed,
one can still \emph{condition} on the assignment to any given
variable to obtain a solution to the $k-1$ round Lasserre's SDP that
corresponds to the \emph{conditioned distribution}.

\ifnum\full=0 \vspace{-8pt}\fi
\section{Globally Uncorrelated SDP Solutions} \label{sec:uncorrelatedsol}

As remarked earlier, it is easy to round SDP solutions to a CSP with
cardinality constraint if the variables behave like {\it independent} random
variables.  In this section, we show a very simple procedure that
starts with a solution to the $(k+l)$-round Lasserre SDP and produces
a solution to the $l$-round Lasserre SDP with the additional property
that globally the variables are somewhat "uncorrelated".  To this end, we define the
notion of {\it $\alpha$-independence} for SDP solutions below.
We remark that all the definitions and results in this section can be applied to all CSPs.
%
%

\begin{definition} \label{def:alphaindependentsol}
Given a solution to the $k$-round Lasserre SDP relaxation, it is said
to be $\alpha$-independent if
$
\E_{i,j\sim W}[I_{\mu_{\{i,j\}}}(X_i;X_j)]\leq \alpha
$
where $\mu_{\{i,j\}}$ is the local distribution associated with the
pair of vertices $\{i,j\}$.
\end{definition}
\ifnum\full=1
\begin{remark}
We stress again that the variables in the SDP solution are not jointly distributed. However, the notion is still well-defined here because of the locality of mutual information: it only depends on the joint distribution of two variables, which is guaranteed to exist by the SDP. Also, $\mu_{\{i,j\}}$ in the expression can be replaced with $\mu_{S}$ for arbitrary $S$ with $i,j\in S$ and $|S|\leq k$ because of the consistency of local distributions.
\end{remark}
\fi

The notion of $\alpha$-independence of random variables using mutual
information, easily translates in to more familiar notion of
statistical distance.  Specifically, we have the following relation.
\begin{fact} \label{fact:statdist}
Let $X$ and $Y$ be two jointly distributed random variables on $[q]$
then,
$$ I(X;Y)\geq \frac{1}{2\ln2}\sum_{i,j\in
[q]}(\mathbb{P}(X=i,Y=j)-\mathbb{P}(X=i)\mathbb{P}(Y=j))^2 \mcom $$
in particular for all $i,j\in [q]$
$$
|\mathbb{P}(X=i,Y=j)-\mathbb{P}(X=i)\mathbb{P}(Y=j)|\leq \sqrt{2 I(X;Y)}
$$
As a consequence, if $X$ and $Y$ are two random variables defined on $\{-1,1\}$,
$
\mbox{Cov}(X,Y)\leq O(\sqrt{I(X;Y)})
$
\end{fact}
\ifnum\full=1
For the sake of completeness, we include the proof of this observation
in \pref{app:informationtheory}.
\fi
Now we describe the procedure of getting an $\alpha$-independent
$l$-rounds Lasserre's solution.  A similar argument was concurrently discovered in \cite{BarakRS11}.
Here we reproduce the argument in information theoretic terms, while
\cite{BarakRS11} present the argument in terms of covariance.  The
information theoretic argument is somewhat robust and cleaner in that it is
independent of the sample space involved.

\begin{mybox}
\begin{algorithm}

{\bf Input}: A feasible solution to the $(k+l)$ round Lasserre SDP
relaxation as described in \pref{sec:prelim} for $k=1/\sqrt{\alpha}$.

{\bf Output}: An $\alpha$-independent solution to the $l$ round Lasserre SDP
relaxation.\\

Sample indices ${i_1},\ldots,{i_{k}} \subseteq V$ independently according to $W$.  Set $t = 1$.

Until the SDP solution is $\alpha$-independent repeat
\begin{itemize}\itemsep=0ex
	\item	Sample the variable $X_{i_t}$ from its marginal distribution after the
	first $t-1$ fixings, and condition the SDP solution on the outcome.
\item $t = t+1$.
\end{itemize}
\end{algorithm}
\end{mybox}

The following lemma shows that there exists $t$ such that the
resulting solution is $\alpha$-independent after $t$-conditionings with high probability.
\begin{lemma} \label{lem:existlowmi}
There exists $t\leq k$ such that $\E_{i_1,\ldots,i_t \sim
W}\E_{i,j\sim W}[I(X_i,X_j|X_{i_1},\ldots,X_{i_{t}})]\leq \frac{\log q}{k-1}$
\end{lemma}

\begin{proof}
By linearity of expectation, we have that for any $t\leq k-2$
$$
\E_{i,i_1,\ldots,i_t \sim W}[H(X_i|X_{i_1},\ldots,X_{i_t})]=
\E_{i,i_1,\ldots,i_t \sim W} [H(X_i|X_{i_1},\ldots,X_{i_{t-1}})]-
\E_{i_1,\ldots,i_{t-1} \sim W}\E_{i,i_t \sim W}[
I(X_i,X_{i_t}|X_{i_1},\ldots,X_{i_{t-1}})]
$$
adding the equalities from $t=1$ to $t=k-2$, we get
$$
\E_{i \sim W}[H(X_i)]-\E_{i_{1},\ldots,i_{k-2}\sim
W}[H(X_i|X_{i_1},\ldots,X_{i_{k-2}})]=\sum_{1\leq t\leq
k-1}\E_{i, j, i_1,\ldots,i_{t-1} \sim
W}[I(X_i,X_j|X_{i_1},\ldots,X_{i_{t-1}})]
$$
The lemma follows from the fact that for each $i$, $H(X_i)\leq \log q$.
\end{proof}

\begin{theorem}  \label{thm:alphaindependentsol}
For every $\alpha>0$ and positive integer $\ell$, there exists an algorithm running in time
$O(n^{poly(1/\alpha)+\ell})$ that finds an
$\alpha$-independent solution to the $\ell$-round Lasserre SDP,
with an SDP objective value of at least $\OPT-\alpha$, where $\OPT$
denotes the optimum value of the $\ell$-round Lasserre SDP relaxation.
\end{theorem}

\begin{proof}
Pick $k=\frac{4\log q}{\alpha^2}$.  Solve the $k+\ell$ round Lasserre
SDP solution, and use it as input to the conditioning algorithm
described earlier.  Notice that the algorithm respects the marginal
distributions provided by the SDP while sampling the values to
variables.  Therefore, the expected objective value of the SDP
solution after conditioning is exactly equal to the SDP objective
value before conditioning. Also notice that the SDP value is at most
$1$. Therefore, the probability of the SDP value dropping by at least
$\alpha$ due to conditioning is at most $1/(1+\alpha)$.

Also, by \pref{lem:existlowmi} and Markov Inequality, the probability
of the algorithm failing to find a $\sqrt{\frac{\log
q}{k}}$-independent soluton is at most $\sqrt{\frac{\log q}{k}}$.
Therefore, by union bound, there exists a fixing such that the SDP
value is maintained up to $\alpha$, and the solution after
conditioning is $\alpha$-independent. Moreover, this particular fixing
can be found using brute-force search.

\end{proof}

\ifnum\full=0 \vspace{-8pt}\fi
\section{Rounding Scheme for \maxbisection} \label{sec:rounding}

In this section, we present and analyze a natural rounding scheme for
\maxbisection.  Given an globally uncorrelated SDP solution to a
$2$-round Lasserre SDP relaxation of \maxbisection, the rounding
scheme will output a cut with the approximation guarantees outlined in
\pref{thm:max-bisection}.
The same rounding scheme also yields a $0.92$-approximation algorithm for arbitrary globally constrained \maxtwosat problem.

\paragraph{Constructing Goemans-Williamson type SDP solution}
In the $2$-round Lasserre SDP for \maxbisection, there are two orthogonal vectors
$v_{i0}$ and $v_{i1}$ for each variable $x_i$.  This can be used to
obtain a solution to the Goemans-Williamson SDP solution by simply defining $v_i
\defeq v_{i0} - v_{i1}$.  The following proposition is an easy consequence,
\begin{proposition}
Let $v_i=v_{i0}-v_{i1}=(2p_i-1)I+w_i$ where $p_i=\Pr(x_i=0)$. Then, for each edge $e=(i,j)\in E$, $\mathbb{P}_{\mu_e}(x_i\neq x_j)=\|v_i-v_j\|^2/4$.
\end{proposition}
\ifnum\full=1
\begin{proof}
$$
\|v_i-v_j\|^2=2-2\langle v_{i0}-v_{i1},v_{j0}-v_{j1} \rangle = 2-2(\mathbb{P}_{\mu_e}(x_i=x_j)-\mathbb{P}_{\mu_e}(x_i\neq x_j))=4\mathbb{P}_{\mu_e}(x_i\neq x_j)
$$
\end{proof}
\fi
Let $w_i$ be the component of $v_{i}$ orthogonal to the $I$ vector, i.e.,
$ w_i \defeq (v_{i} - \iprod{v_{i},I}I) \mper $
Using $v_{i0} + v_{i1} = I$ and $\iprod{v_{i0},v_{i1}} = 0$, we get $v_{i0} = \iprod{v_{i0},I} I +
w_i/2$ and $v_{i1} = \iprod{v_{i1},I} I - w_i/2$.
%
%
%
%
%
%
%
%
%
%
We remark that $w_i$ is the crucial component that captures the \emph{correlation} between $x_i$ and other variables. To formalize this, we show the following lemma.
\begin{lemma} \label{lem:ipbound}
Let $v_i$ and $v_j$ be the unit vectors constructed above, $w_i$ and $w_j$ be the components of $v_i$ and $v_j$ that orthogonal to $I$. Then
$
|\langle w_i, w_j \rangle|\leq 4\sqrt{2I(x_i,x_j)}
$
\end{lemma}
\begin{proof}  Let $p_i \defeq \Pr(x_i = 0) = \iprod{v_{i0},I}$ and
$p_j \defeq \Pr(x_j = 0) = \iprod{v_{j0},I}$.  Notice that
$$
|\Pr(x_i=0, x_j=0)-\Pr(x_i=0)\Pr(x_j=0)|=\|\langle p_i I+w_i/2, p_j I+w_j/2 \rangle - p_i p_j\|= |\langle w_i,w_j \rangle|/4
$$
By applying \pref{fact:statdist}, we get
$
|\langle w_i,w_j \rangle|\leq 4\sqrt{2I(x_i;x_j)}
$
\end{proof}
Henceforth we will switch from the alphabet  $\{0,1\}$ to $\{-1,1\}$
\footnote{The mapping is given by $0 \rightarrow 1$ and $1 \rightarrow
-1$}. After this transformation, we can interpret the inner product $\mu_i=\langle v_i,I\rangle=p_i-(1-p_i)$ as the \emph{bias} of vertex $i$.
\ifnum\full=0 \vspace{-8pt}\fi
\subsection{Rounding Scheme}
\ifnum\full=0 \vspace{-8pt}\fi

Roughly speaking, the algorithm applies a hyperplane rounding on
the vectors $w_i = v_i - \iprod{v_i,I}I$ associated with the vertices
$i \in V$.  However, for each vertex $i \in V$, the algorithm shifts
the hyperplane according to the bias of that vertex.

\begin{mybox}
\begin{algorithm} \label{alg:rounding}
Given: A set of unit vectors $\{v_1,\ldots,v_n\}$ where $v_i=\mu_i I+w_i$, where $w_i$ is the component of $v_i$ orthogonal to $I$.

Pick a random Gaussian vector $g$ orthogonal to $I$ with coordinates distributed as $\mathcal{N}(0,1)$.
For every $i$,
\begin{enumerate}
\item
Project $g$ on the direction of $w_i$, \ie
$
\xi_i=\langle g,\bar{w_i} \rangle
$,
where $\bar{w}_i=\frac{w_i}{\sqrt{1-\mu_i^2}}$ is the normalized vector or $w_i$. Note that $\xi_i$ is also a standard Gaussian variable.
\item
Pick threshold $t_i$ as follows:
$$
t_i=\Phi^{-1}(\mu_i/2+1/2)
$$

\item
If $\xi_i\leq t_i$, set $x_i=1$, otherwise set $x_i=-1$.
\end{enumerate}
\end{algorithm}
\end{mybox}

Notice that, the threshold $t_i$ is chosen so that individually the bias of $x_i$ is exactly $\mu_i$. Therefore, the expected balance of the rounded solution matches the intended value.
The analysis of the rounding algorithm consists of two parts: first we
show that the cut returned by the rounding algorithm has high expected
value, then we show the that the balance of the cut is concentrated
around its expectation.

\ifnum\full=0 \vspace{-8pt}\fi
\subsection{Analysis of the Cut Value}
\ifnum\full=0 \vspace{-8pt}\fi
Analyzing the cut value of the rounding scheme is fairly standard
albeit a bit technical.  The analysis is {\it local} as in the case of
other algorithms for CSPs, and reduces to bounding the probability
that a given edge is cut.  The probability that a given edge $u,v$ is
cut corresponds to a probability of an event related to two correlated
Gaussians.

By using numerical techniques, we were able to show that the
cut value is at least $0.85$ times the SDP optimum.  Analytically, we
show the following asymptotic relation.
\begin{lemma} \label{lem:rooteps}
Let $u=\mu_1 I+ w_1$,$v=\mu_2 I + w_2$ be two unit vectors satisfying $\|u-v\|^2/4\leq \eps$, then the probability of them being separated by \pref{alg:rounding} is at most $O(\sqrt{\eps})$.
\end{lemma}
\ifnum\full=1
The proof of this lemma is fairly technical and is deferred to \pref{app:cutvalue}.
\fi
\ifnum\full=0
The proof of this lemma is fairly technical and is deferred to the
full version.
\fi
\ifnum\full=0 \vspace{-8pt}\fi
\subsection{Analysis of the Balance}
\ifnum\full=0 \vspace{-8pt}\fi
In this section we show that the balance of the rounded solution will be highly concentrated. We prove this fact by bounding the variance of the balance. Specifically, we show that if the SDP solution is $\alpha$-independent, then the variance of the balance can be bounded above by a function of $\alpha$.

The proof in this section is information theoretical -- although this approach gives sub-optimal bound, but the proof itself is very simple and clean.

\begin{lemma} \label{lem:lowmi}
Let $v_i=\mu_i I+w_i$ and $v_j=\mu_j I+w_j$ be two vectors in the SDP solution that satisfy $|\langle w_i, w_j\rangle|\leq \zeta$. Let $y_i$ and $y_j$ be the rounded solution of $v_i$ and $v_j$, then
$
I(y_i;y_j)\leq O(\zeta^{1/3})
$

\end{lemma}

\begin{proof}
Since
$$
|\langle w_i,w_j \rangle |=\sqrt{1-\mu_i^2}\sqrt{1-\mu_j^2}|\langle \bar{w}_i,  \bar{w}_j\rangle |\leq \zeta
$$

It implies that one of the three quantities in the equation above is at most $\zeta^{1/3}$.
If it is the case that $\sqrt{1-\mu_i^2}\leq \zeta^{1/3}$ or $\sqrt{1-\mu_j^2}\leq \zeta^{1/3}$ (w.l.o.g we can assume it's the first case), then we have
$$
\min(|1-\mu_i|,|1+\mu_i|)\leq O(\zeta^{2/3})
$$
We may assume $\mu_i>0$, therefore $1-\mu_i<O(\zeta^{2/3})$.
Notice that our rounding scheme preserves the bias individually, which implies $y_i$ is a highly biased binary variable, hence
$$
I(y_i,y_j)\leq H(y_i) = O(-(1-\mu_i)\log(1-\mu_i))\leq O(\zeta^{1/3})
$$
Now let's assume it's the case that $|\langle \bar{w}_i,  \bar{w}_j\rangle |\leq \zeta^{1/3}$.
Let $g_1=g\cdot \bar{w}_1$ and $g_2=g\cdot \bar{w}_2$ as described in the rounding scheme, and $\rho=\langle \bar{w}_i,  \bar{w}_j\rangle$. Hence $g_1$ and $g_2$ are two jointly distributed
standard Gaussian variables with covariance matrix  $\Sigma=
\begin{pmatrix}
1\ \ \ \rho \\ \rho\ \ \ 1\\
\end{pmatrix}$.

The mutual information of $g_1$ and $g_2$ is
$$
I(g_1,g_2)=-\frac{1}{2}\log(\det\Sigma)\leq O(-\log(1-\zeta^{2/3}))\leq O(\zeta^{1/3})
$$

Notice that $y_i$ is fully dependent on $g_i$, therefore by the data
processing inequality (\pref{thm:dataprocessing}), we have $I(y_1,y_2)\leq I(g_1,g_2)\leq O(\zeta^{1/3})$
\end{proof}

\begin{theorem}
Given an $\alpha$-independent solution to 2-rounds Lasserre's SDP hierarchy. Let $\{y_i\}$ be the rounded solution after applying \pref{alg:rounding}. Define $S=\E_{i\sim W} y_i$, then

$$
\mbox{Var}(S)\leq O(\alpha^{1/12})
$$
\end{theorem}

\begin{proof}
\begin{align*}
\mbox{Var}(S) &= \E_{i,j\sim W}[\mbox{Cov}(y_i,y_j)] \\
&\leq \E_{i,j\sim W}[O(\sqrt{I(y_i;y_j)})]  \ \ \ \ \ \text{(by
\pref{fact:statdist})}\\
&\leq \E_{i,j\sim W}[O(\sqrt{|w_i,w_j|^{1/3}})] \ \ \ \ \text{(by \pref{lem:lowmi})}  \ \\
&\leq \E_{i,j\sim W}[O(\sqrt{I(x_i;x_j)^{1/6}})] \ \ \ \ \text{(by \pref{lem:ipbound})}  \ \\
&\leq O((\E_{i,j\sim W}[I(x_i;x_j)])^{1/12})  \ \ \ \ \ \text{(by concavity of the function $x^{1/12}$)} \\
&\leq O(\alpha^{1/12})
\end{align*}
\end{proof}

\begin{corollary} \label{cor:balance}
Given an $\alpha$-independent solution to 2-rounds Lasserre's SDP hierarchy $v_i=\mu_i+w_i$. The rounding algorithm will find an $O(\alpha^{1/24})$-balanced (that is, the balance of the cut differs from the expected value by at most $O(\alpha^{1/24})$ fraction of the total weights) with probability at least $1-O(\alpha^{1/24})$.
\end{corollary}

\ifnum\full=0 \vspace{-8pt}\fi
\subsection {Wrapping Up}
\ifnum\full=0 \vspace{-8pt}\fi
Here we present the proofs of the main theorems of this work.
\paragraph{Proof of \pref{thm:minbisection}}
Suppose we're given a \minbisection instance $G=(V,E)$ with value at
most $\epsilon$ and constant $\delta>0$. By setting
$\alpha=\delta^{24}$ and applying \pref{thm:alphaindependentsol}, we
will get an $\alpha$-independent solution with value at most
$\epsilon+\alpha$. By \pref{lem:rooteps} and the concavity of the
function $\sqrt{x}$, the expected size of the cut returned by
\pref{alg:rounding} is at most
$O(\sqrt{\eps+\alpha})=O(\sqrt{\eps}+\sqrt{\alpha})$. Therefore, with
constant probability (say 1/2), the cut returned by the rounding
algorithm has size at most $O(\sqrt{\eps}+\sqrt{\alpha})$. Also, by
\pref{cor:balance}, the cut will be $O(\delta)$-balanced with
probability at least $1-O(\delta)$. Therefore, by union bound, the
algorithm will return an $O(\delta)$-balanced cut with value at most
$O(\sqrt{\eps}+\sqrt{\alpha})$ with constant probability. Notice that
this probability can be amplified to $1-\eps$ by running the algorithm
$O(\log(1/\eps))$ times. Given such a cut, we can simply move $O(\delta)$ fraction of the vertices with least degree from the larger side to the smaller side to get an exact bisection -- this process will increase the value of the cut by at most $O(\delta)$. Therefore, in this case, we get a bisection of value at most $O(\sqrt{\eps}+\sqrt{\alpha}+\delta)=O(\sqrt{\eps}+\delta)$. Hence, the expected value of the bisection returned by the rounding algorithm is at most $(1-\eps)O(\sqrt{\eps}+\delta)+\eps=O(\sqrt{\eps}+\delta)$.
\paragraph{Proof of \pref{thm:max-bisection}}
The proof is similar in the case of \maxbisection. The only difference is that we have to use the fact that the rounding scheme is balanced, \ie  $\Pr(F(v)\neq F(-v))=1$. Hence, by \pref{lem:rooteps}, for any edge $(u,v)$ with value $1-\epsilon$ in the SDP solution, the algorithm separates them with probability at least $1-O(\sqrt{\eps})$. The rest of the proof is identical.

Using a computer-assisted proof, we can show that the approximation
ratio of this algorithm for \maxbisection is between $0.85$ and $0.86$.
Thus further narrowing down the gap between approximation and inapproximability of \maxbisection.
Using the same algorithm, we obtain a $0.92$-approximation for globally constrained \maxtwosat. It is known that under the Unique Games Conjecture, \maxtwosat is NP-Hard to approximate within $0.9401$.

\ifnum\full=1

\section{Dictatorship Tests from Globally Uncorrelated SDP Solutions}

\label{sec:gaptodict}

A dictatorship test $\dict$ for the \maxbisection problem
	consists of a graph on the set of vertices
	$\sbits^{R}$.  By convention, the graph $\dict$
	is a weighted graph where the edge weights form a
	probability distribution (sum up to $1$).  We will
	write $(\mrv{z},\mrv{z}') \in \dict$ to denote an edge
	sampled from the graph $\dict$ (here $\mrv{z},\mrv{z}'
	\in \sbits^{R}$).

A cut of the $\dict$ graph can be thought of as a boolean function
	$\cF : \sbits^R \to \sbits$.  The value	of a cut $\cF$ given by
		$$ \dict(\cF) = \frac{1}{2}\E_{ (\mrv{z}, \mrv{z}')\in
\dict} \Big[
		1 - \cF(\mrv{z}) \cF(\mrv{z}') \Big] \mcom$$
	is the probability that $\mrv{z}$,$\mrv{z}'$ are on
different sides of the cut.
		It is also useful to define $\dict(\cF)$ for non-boolean
	functions $\cF: \sbits^R \to [-1,1]$ that take values
	in the interval $[-1,1]$.  To this end, we will interpret a
	value $\cF(\mrv{z}) \in [-1,1]$ as a random variable that
	takes $\sbits$ values.  Specifically, we think of a number $a
	\in [-1,1]$ as the following random variable
	\begin{align} \label{eq:cutrounding}
	a = \begin{cases} -1 & \text{ with probability }
		\frac{1-a}{2} \\
		1 & \text{ with probability } \frac{1+a}{2}
		\end{cases}
	\end{align}
	With this interpretation, the natural definition of
	$\dict(\cF)$ for such a function is as follows:
	$$  \dict(\cF) = \frac{1}{2}\E_{ (\mrv{z}, \mrv{z}') \in
	\dict} \Big[
		1 - \cF(\mrv{z}) \cF(\mrv{z}') \Big] \mper $$
	Indeed, the above expression is equal to the expected value of the
	cut obtained by randomly rounding the values of the function
	$\cF : \sbits^{R} \to [-1,1]$ to $\sbits$ as
	described in Equation \eqref{eq:cutrounding}.

	We will construct a dictatorship test for the weighted version
	of \maxbisection.  In particular, each vertex $x \in \sbits^R$ of \dict is
	associated a weight $W(x)$, and the weights $W$ form a
	probability distribution over $\sbits^R$ (sum up to $1$).  The
	balance condition on the cut can now be expressed as
	$\E_{\mrv{z}
	\sim W}[\cF(\mrv{z})] = 0$.

The dictatorship test $\dict$ can be easily
	transformed in to a dictatorship test $\dict'$ for unweighted
	\maxbisection.  The idea is to replace each vertex $\mrv{x} \in
	\sbits^R$ with a cluster $V_{\mrv{x}}$ of $\lfloor W(\mrv{x}) \cdot M \rfloor$
	vertices for some large integer $M$.  For every edge
	$(\mrv{x},\mrv{y})$ in $\dict$, connect every pair of vertices
	in the corresponding clusters $V_{\mrv{x}},V_{\mrv{y}}$ with
	edge of the same weight.  Given any bisection $\cF' : \dict'
	\to \sbits$ of the graph $\dict'$
	with value $c$, define $\cF(\mrv{z}) = \E_{v \in V_{\mrv{z}}}
	\cF'(v)$.  By slightly correcting the balance of $\cF$, it is
	easy to obtain a bisection $\cF : \sbits^R \to [-1,1]$
	satisfying
	$$ \dict(\cF) \geq c - o_{M}(1)  \qquad \qquad \E_{\mrv{z}}
	\cF(\mrv{z}) = 0 \mper$$
	Conversely, given a bisection $\cF : \sbits^R \to [-1,1]$ of
	$\dict$, assign $(1+\cF(\mrv{z}))/2$ fraction of vertices of
	$V_{\mrv{z}}$ to be $1$ and the rest to $-1$.  The resulting
	partition of $\dict'$ is very close to balanced (up to
	rounding errors), and can be modified in to a bisection with
	value $\dict(\cF) - o_{M}(1)$.

	The {\it dictator cuts} are given by the functions
	$\cF(\mrv{z}) = z^{(\ell)}$ for some $\ell \in [R]$.
	The dictatorship test graph is so constructed that each
	dictator cut will yield a bisection and the $\mathsf{Completeness}$ of the test $\dict$ is the minimum value of a dictator cut, i.e.,
	$$ \mathsf{Completeness}(\dict) = \min_{\ell \in [R]}
	\dict(z^{(\ell)}) $$
	The soundness of the dictatorship test is the value of
	bisections of $\dict$ that are {\it far from every dictator}.  We will formalize the notion of being
	{\it far from every dictator} using the notion of influences.

\paragraph{Influences and Noise Operators}

	To this end, we recall the definitions of influences and noise
	operators.  Let $\Omega = (\sbits,\mu)$ denote the probability space with atoms
$\sbits$ and a distribution $\mu$ on them.  Then, the influences and
noise operators for functions over the product space $\Omega^R$ are
defined as follows.
	\begin{definition}[Influences]
	The {\it influence} of the $\ell$\th coordinate on a function
	$\orf{F}: \sbits^R \to \R$ under a distribution $\mu$
	over $\sbits$ is given by
	$\Inf^{\mu}_\ell(\orf{F})
	=\E_{\mrv{x}^{(-\ell)}}\big[\Var_{x^{(\ell)}}[\orf{F}(\mrv{x})]\big]
	= \sum_{S \ni \ell} \hat{\orf{F}}_S^2$.
	\end{definition}
	
\begin{definition}
  For $0\leq \eps \leq 1$, define the operator $\T_{1-\eps}$ on
$L_{2}(\Omega^{R})$ as,
$$ \T_{1-\eps} \orf{F} (\mrv{z})  = \E[\orf{F}(\tilde{\mrv{z}})\mid
\mrv{z}] $$
where each coordinate $\tilde{z}^{(i)}$ of $\tilde{\mrv{z}}$ is  equal
to $z^{(i)}$
\textrm{ with probability $1-\eps$} and a random element from
$\Omega$ with probability $\eps$.
\end{definition}

\paragraph{Invariance Principle}

The following invariance principle is an immediate consequence of
Theorem $3.6$ in the work of Isaksson and Mossel \cite{IsakssonM09}.

\begin{theorem}(Invariance Principle \cite{IsakssonM09}) \label{thm:invariance}
  Let $\Omega$ be a finite probability space with the least non-zero
  probability of an atom at least $\alpha \leq 1/2$.  Let $\mcl{L} =
  \{\ell_1,\ell_2\}$ be an ensemble of random variables
  over $\Omega$.  Let $\erv{G} =
  \{g_1,g_2\}$ be an ensemble of Gaussian random variables satisfying the following conditions:
  \begin{align*}
    \E[\ell_i] = \E[g_i]  & & \E[\ell_i^2] = \E[g_i^2] & & \E[\ell_i
    \ell_j] = \E[g_i g_j] & & \forall i,j \in \{1,2\}
  \end{align*}
  Let $K = \log (1/\alpha)$.  Let $\opl{F}$ denote a
  multilinear polynomial  and let $\opl{H} = (T_{1-\epsilon}
  \opl{F})$.  Let the variance of $\opl{H}$, $\Var[\opl{H}]$ be
  bounded by $1$ and all the influences are smaller than $\tau$, i.e.,  $\Inf_i(\opl{H})
  \leq \tau$ for all $i$.

If $\Psi : \R^2 \rightarrow \R$ is a Lipschitz-continous function with
Lipschitz constant $C_0$ (with respect to the $L_2$ norm)  then
    $$ \Big|\E\Big[\Psi(\opl{H}(\mcl{L}^{R}))\Big] -
    \E\Big[\Psi(\opl{H}(\erv{G}^R))\Big] \Big| \leq
    C \cdot C_0 \cdot \tau^{\epsilon/18K} = o_{\tau}(1) $$
    for some constant $C$.
\end{theorem}

\paragraph{Construction}

	Let $G = (V,E)$ be an arbitrary instance of \maxbisection.  Let
	$\vec{V} = \{v_{i,0},v_{i,1}\}_{i \in V}$ denote a {\it globally
	uncorrelated} feasible SDP solution for two rounds of the
	Lasserre hierarchy.  Specifically, for every pair of vertices
	$i,j \in V$, there exists a distribution $\mu_{ij}$ over $\sbits$ assignments that match the SDP inner
products.  In other words, there exists $\sbits$
valued random variables $z_i,z_j$ such that
$$ \iprod{\vec v_i , \vec v_j} = \E[z_i \cdot z_j] \mper$$
Furthermore, the correlation between random pair of vertices is at
most $\delta$, i.e., $$ \E_{i,j \in V}[I(z_i,z_j)] \leq
\delta \mper$$

	Starting from $G = (V,E)$ along with the SDP solution $\vec{V}$ and a parameter $\epsilon$ we
	construct a dictatorship test $\dict_{\vec{V}}^{\eps}$.  The
	dictatorship test gadget is exactly the same as
	the construction by Raghavendra \cite{Raghavendra08} for the \maxcut
	problem.  For the sake of completeness, we include the details
	below.

	\begin{mybox}
$\dict^{\eps}_{\vec{V}}$ (\maxbisection)
\ifnum\full=0
\itemsep=0ex
\fi
The set of vertices of $\dict^{\eps}_{\vec V}$ consists of the
$R$-dimensional hypercube $\sbits^{R}$.  The distribution of edges in
$\dict^{\eps}_{\vec V}$ is the one induced by the following sampling
procedure:
\begin{itemize} \itemsep=0ex
\item Sample an edge $e = (v_i,v_j) \in E$ in the graph $G$.
\item Sample $R$ times independently from the distribution $\mu_{e}$ to
	obtain $\mrv{z}_i^R = (z_i^{(1)},\ldots, z_{i}^{(R)})$ and
	$\mrv{z}^R_j = (z_j^{(1)},\ldots, z_{j}^{(R)})$, both in
	$\sbits^{R}$.
\item  Perturb each coordinate of $\mrv{z}^R_i$ and $\mrv{z}^R_j$
	independently with probability $\epsilon$ to obtain
	$\tilde{\mrv{z}}^R_i,\tilde{\mrv{z}}^R_j$
	respectively.  Formally, for each $\ell \in [R]$,
	$$ \tilde{z}_i^{(\ell)} = \begin{cases} z_{i}^{(\ell)} & \text{
		with probability } 1-\eps \\
		\text{ random sample from distribution } \mu_i  & \text{
		with probability } \eps
	\end{cases} $$
\item Output the edge $(\tilde{\mrv{z}}^{R}_i,
	\tilde{\mrv{z}}^{R}_j)$. 
\end{itemize}
The weights on the vertices of $\dict^{\eps}_{\vec V}$ is given by
$$ W(x) = \E_{i \in V}\left[ \Pr_{\mrv{z} \in \mu_i^R}[\mrv{z} =
x]\right] \mper $$
\end{mybox}

  We will show the following theorem about the completeness and
soundness of the dictatorship test.

\begin{theorem} \label{thm:gaptodict}
There exist absolute constants $C,K$ such that for all $\eps, \tau
\in [0,1]$ there exists $\delta$ such that following holds. Given a graph $G$ and a
$\delta$-independent SDP solution $\vec V = \{\vec
v_{i,0}, \vec v_{i,1}|i \in V\}$  for the two round Lasserre SDP for \maxbisection, the
dictatorship test $\dict_{\vec V}^{\eps}$ is such that
\begin{itemize}
	\item The {\it dictator cuts} are bisections with value within
		$2\epsilon$ of the SDP value, i.e.,
		$\msf{Completeness}(\dict^{\eps}_{\vec V}) \geq \val(\vec V) -
		2\eps$
	\item If $\cF: \sbits^R \to [-1,1]$ is a bisection of
		$\dict_{\vec V}^{\eps}$ ($\E_{x \sim W}[\cF(x)]
		= 0$) and all its influences are at most $\tau$, i.e.,
		$$\Inf_\ell^{\mu_i}(\cF) \leq \tau \qquad\qquad
		\forall i \in V, \ell \in [R]\mcom$$ then,
		$$ \dict_{\vec V}^{\eps}(\cF) \leq \opt(G)+
		C\tau^{K\eps} \mper $$
\end{itemize}
\end{theorem}

\begin{proof}
The analysis of the dictatorship test is along the lines of the
corresponding proof for \maxcut in \cite{Raghavendra08}.

\paragraph{Completeness}
First, the dictatorship test gadget is exactly the same as that
constructed for \maxcut in \cite{Raghavendra08}.  Therefore from \cite{Raghavendra08}, the
fraction of edges cut by the dictators is at least $\val(\vec V) - 2\epsilon$.
To finish the proof of completeness, we need to show that the dictator
cuts are indeed {\it balanced}.  However, this is an easy calculation
since the balance of the $j^{th}$ dictator cut is given by,
$$ \E_{x \in W} [x^{(j)}] = \E_{i \in V} \E_{x \in \mu_i^R}[x^{(j)}] =
\E_{i \in V} \E_{a \in \mu_i}[a] = 0 \mcom$$
where the last equality uses the fact that the SDP solution satisfies
the balance condition.

\paragraph{Soundness}
Let $\cF : \sbits^R \to [-1,1]$ be a balanced cut all of whose
influences are at most $\tau$.  As in \cite{Raghavendra08}, we will
use the function $\mrf{F}$ to round the SDP solution $\vec V$.  The
rounding algorithm is exactly the same as the one in
\cite{Raghavendra08}.  For the sake of completeness, we reproduce the
rounding scheme below.

\begin{mybox}
      $\round_{\mrf{F}}$ Scheme
      \paragraph{Truncation Function} Let $\struncate : \R \to [-1,1]$ be a Lipschitz-continous
function such that for all $x \in [-1,1]$, $\struncate(x)
= x$.  Let $C_0$ denote the Lipschitz constant of the function
$\struncate$.
%

\paragraph{Bias}  For each vertex $i \in V$, let the bias of vertex
$i$ be $\theta_i = \iprod{\vec v_{i,0}, \vec I}$ and let $\vec w_i
=\vec v_{i,0} -
\iprod{\vec v_{i,0},\vec I} \vec v_{i,0}$ be the component of $\vec v_{i,0}$ orthogonal to the vector
$\vec I$.
\paragraph{Scheme}
	Sample $R$ vectors $\zeta^{(1)},\ldots,\zeta^{(R)}$ with each coordinate
      being i.i.d normal random variable.

      For each $i \in V$ do
      \begin{itemize}	
        \itemsep=0ex
      \item For all $1 \leq  j \leq R $,
        compute the projection $g_{i}^{(j)}$ of the vector
	$\vec w_i$ as follows:
        \begin{align*}
          {g}_{i}^{(j)}  = \theta_i +
          \Big[\iprod{\vec w_i ,
          \zeta^{(j)}}\Big] 	
        \end{align*}
	and 
	  let $\mrv{g}_i =  (g_i^{(1)},\ldots,g_i^{(R)})$

\item   Let $\opl{F}_i$ denote the multilinear polynomial
	corresponding to the function $\mrf{F}$ under the distribution
	$\mu_i^R$ and let $\opl{H}_i = \T_{1-\epsilon} \opl{F}_i$.
	Evaluate $\opl{H}_i$ with $g_{i}^{(j)}$ as
	inputs to obtain $p_i$, i.e., $p_i =
	\opl{H}_i(g_{i}^{(1)},\ldots,g_{i}^{(R)})$.
      \item 	Round $p_i$ to $p_i^{*} \in
        [-1,1]$ by using the Lipschitz-continous truncation function
	$\struncate : \R \to [-1,1]$.
	$$ p_i^* = \struncate(p_i) \mper$$
      \item   Assign the vertex $i$ to be $1$ with probability
	      $(1+p_i^{*})/2$ and $-1$ with the remaining probability.
      \end{itemize}
 \end{mybox}

  Let $\round_{\mrf{F}}(\vec V)$ denote the
  expected value of the cut returned by the rounding scheme
$\round_{\mrf{F}}$ on the SDP solution $\vec V$ for the
\maxbisection instance $G$.

Again, by appealing to the soundness analysis in \cite{Raghavendra08},
we conclude that the fraction of edges cut by the resulting partition
is lower bounded by
$$ \round_{\mrf{F}}(\vec V) \geq \dict_{\vec V}^{\eps}(\cF) -
C'\tau^{K\eps} \mper$$
for an absolute constant $C'$.  To finish the proof, we need to argue that if the SDP solution $\vec
V$ is $\delta$-independent, then the resulting partition is close to
balanced with high probability.

First, note that the expected balance of the cut is given by,
$$ \E_{\zeta} \left[ \E_{i}[ p_i^*] \right] =  \E_{\zeta} \left[
\E_{i}[\struncate(\opl{H}(g_{i})) ] \right] \mper $$

Fix a vertex $i \in V$. By construction, the random variables
$z^{(\ell)}_i \sim \mu_i$ and $g^{(\ell)}_i$ have
matching moments up to order two for each $\ell \in [R]$.  Therefore, by applying the invariance
principle of Isaksson and Mossel \cite{IsakssonM09} with the smooth
function $\struncate$ and the multilinear polynomial $\opl{F}_i$ yields
the following inequality,
$$ \E_{\zeta} \left[
\struncate(\opl{H}_i(g_{i}))  \right] \leq  \E_{\mrv{z}^R_i \in
\mu_i^R} \left[\struncate(\opl{H}_i(\mrv{z}_i^R)) \right] +
C\tau^{K\eps} \mper$$
Since the cut $\mrf{F}$ is balanced we can write,
$$ \E_{i} \E_{\mrv{z}^R_i \in \mu_i^R}
\left[\struncate(\opl{H}_i(\mrv{z}_i^R)) \right] =  \E_{i}
\E_{\mrv{z}_i^R \in \mu_i^R}
\left[\opl{H}_i(\mrv{z}_i^R) \right] = \E_{i} \E_{\mrv{z}_i^R \in \mu_i^R}
\left[\opl{F}_i(\mrv{z}_i^R) \right] = \E_{i} \E_{\mrv{z}_i^R \in
\mu_i^R} \left[\mrf{F}(\mrv{z}_i^R) \right] =
0 \mper$$
In the previous calculation, the first equality uses the fact that $\struncate(x) = x$ for $x \in
[-1,1]$ while the second equality uses the fact that
$\E_{\mrv{z}}[\T_{1-\epsilon} H_i(\mrv{z})] = \E_{\mrv{z}}[F_i(\mrv{z})]$.
Therefore, we get the following bound on the expected value of the
balance of the cut, $ \E_{\zeta} \left[
\struncate(\opl{H}_i(g_{i}))  \right] \leq  C\tau^{K\eps} \mper$

Finally, we will show that the balance of the cut is concentrated
around its expectation.  To this end, we first show the following
continuity of the rounding algorithm.

\begin{lemma} \label{lem:rounding-cont}
	For each $i \in V$ and any vector $\vec w'_i$ satisfying
	$\norm{\vec w'_i}_2 = \norm{\vec w_i}_2$, if $p'_i$ denotes the output
	of the rounding scheme $\round_{\mrf{F}}$ with $\vec w'_i$ instead
	of $\vec w_i$ then,
$$ \| \E_{\zeta}[(p'_i-p^*_i)^2 \| \leq C(R) \norm{\vec w_i - \vec w'_i}^2_2 \mcom$$
for some function of $R$ ($C(R) = 2^{2R}$ suffices).
\end{lemma}
\begin{proof}
	Let $\mrv{g}'_i = (g'^{(1)}_i,\ldots,g'^{(R)}_i)$ denote the projections of the vector $\vec w'_i$ along
the directions $\zeta^{(1)},\zeta^{(2)},\ldots,\zeta^{(R)}$.  The
output of the rounding scheme on $\vec w'_i$ is given by
$p'_i = \struncate(\opl{H}_i(\mrv{g}'_i))$.  Recall that the output of the rounding scheme is given by $p^*_i =
\struncate(\opl{H}_i(\mrv{g}_i))$.

The result is a consequence of the fact that the function
$\struncate\circ \opl{H}_i$ is Lipschitz continous.  Since the
variance of $\mrf{F}(\mrv{z}_i^R)$ is at most $1$, the sum of squares
of coefficients of $\opl{H}_i$ is at most $1$.  Therefore, all the
$2^R$ coefficients of $\opl{H}_i$ are bounded by $1$ in absolute value.

The proof is a simple hybrid argument, where we replace $g^{(\ell)}_i$ by $g'^{(\ell)}_i$ one by
one.  The details of the proof are deferred to the full version.
\end{proof}

\begin{lemma} \label{lem:covariance}
For every $i,j$,
$$ | \E_{\zeta}[p_i^{*}p_j^*] - \E_{\zeta}[p_i^{*}]
\E_{\zeta}[p_j^{*}] | \leq C(R) |\iprod{w_i,w_j}|$$
for some function $C(R)$ of $R$ ($C(R) = 1002^{2R}$ suffices).
\end{lemma}
\begin{proof}
Set $\vec w'_j = \vec w_j - \iprod{\vec w_i,\vec w_j} \frac{\vec
w_i}{\norm{\vec w_i}} +
\iprod{\vec w_i,\vec w_j}\bar{u}$ for a unit vector $\bar{u}$
orthogonal to $\vec w_i$ and $\vec w_j$.  Note that $\vec w'_j$ is orthogonal to $\vec w_i$ and satisfies
$\norm{\vec w_j -
\vec w'_j} \leq 4|\iprod{\vec w_i,\vec w_j}|$.
Let $p'_j$ denote the output of the rounding with $\vec w'_j$ instead of
$\vec w_j$. Since $\vec w'_j$ is orthogonal to $\vec w_i$ all their projections are
independent random variables, which implies that,
$$ \E_{\zeta} [p'_j p^{*}_i] = \E_{\zeta}[p'_j]\E_{\zeta}[p^*_i]
\mper$$.
Moreover, by \pref{lem:rounding-cont} we have,
$$ \E_{\zeta} [(p'_j - p^{*}_j)^2] \leq C(R) \norm{\vec w_j - \vec w'_j}_2^2
\leq C(R) \cdot 16|\iprod{\vec w_i,\vec w_j}|^2 \mper $$.
Combining these inequalities and using Cauchy-Schwartz, we finish the proof as follows,
\begin{align*}
 | \E_{\zeta}[p_i^{*}p_j^*] - \E_{\zeta}[p_i^{*}] \E_{\zeta}[p_j^{*}] |
   & \leq  | \E_{\zeta}[p_i^{*}(p_j^*-p'_j)]|  +| \E_{\zeta}[p_i^{*}]\E_{\zeta}[p'_j - p_j^{*}] |\\
   & \leq  2 \left(\E_{\zeta}[(p'_j -
   p^{*}_j)^2]\right)^{\frac{1}{2}}\left(\E[(p^*_i)^2]\right)^{\frac{1}{2}} \\
& \leq 8C(R) |\iprod{w_i,w_j}|
\end{align*}
%
\end{proof}

To finish the proof, now we bound the variance of the balance of the
cut returned using \pref{lem:covariance}. The variance
of the balance of the cut returned is given by,
$$ \E_{\zeta} (\E_{i}[p^*_i])^2 -(\E_{\zeta} \E_{i}[p^*_i])^2 =
\E_{i,j} \left[\E_{\zeta}[p^*_ip^*_j] -
\E_{\zeta}[p^*_i]\E_{\zeta}[p^*_j] \right] \leq C(R)
\E_{i,j}[|\iprod{w_i,w_j}|] $$
For a $\delta$-independent SDP solution, the above quantity is at most
$C(R) \poly(\delta)$.  This gives the desired result.
\end{proof}

\fi

\ifnum\full=0
\clearpage
\bibliographystyle{abbrv}
\bibliography{refs-groth}
\vspace{2ex}
\Large{{\bf APPENDIX:} {\sc Full version of the paper follows.}}
\includepdf[pages=-]{max-bisection-full.pdf}
\end{document}
\fi

\clearpage
\addreferencesection
\bibliographystyle{amsalpha}
\bibliography{papers}
\clearpage

\ifnum\full=1
\appendix

\section{Analysis of Cut Value} \label{app:cutvalue}

We analyze the rounding algorithm in an indirect way -- first we show that under certain conditions, \pref{alg:rounding} returns a better cut compared to Goemans-Williamson algorithm (in expectation). Then we use an union-bound type argument to give the proof for general cases.

First, we present a bound on the tail of the standard gaussian
distribution.
\begin{lemma} \label{lem:normtail}
For $t\geq 0$,
$$
\Phi^{c}(t)=1-\Phi(t) \leq \frac{\sqrt{2/\pi}e^{-t^2/2}}{t+\sqrt{t^2+8/\pi}}
$$
\end{lemma}

\begin{proof}

We apply the following bound on the error function given in \cite{Komatsu55}
$$
e^{x^2}\int_{x}^{\infty}e^{-y^2}dy \leq \frac{1}{x+\sqrt{x^2+4/\pi}}
$$
by replacing $x$ with $\frac{\sqrt{2}t}{2}$, we get the desired bound.

\end{proof}

From now on, let $\mu_0=\sqrt{1-4/\pi^2}\approx 0.7712$ and $t_0=\Phi^{-1}(\mu_0/2+1/2)\approx 1.2034$.

\begin{lemma} \label{lem:monotonicity}
Let $g(t)=e^{t^2/2}(1-\mu^2(t))$, where $\mu(t)=2\Phi(t)-1$. $g(t)$ is decreasing when $t\geq t_0$.
\end{lemma}

\begin{proof}
By simple calculation, we get
$$
g'(t)=4\left(te^{t^2/2}(1-\Phi(t))\Phi(t)+\frac{1}{\sqrt{2\pi}}(1-2\Phi(t))\right)
$$
we want to show
$$
te^{t^2/2}(1-\Phi(t))\Phi(t)+\frac{1}{\sqrt{2\pi}}(1-2\Phi(t))<0
$$
by applying \pref{lem:normtail}, we only need to show
$$
te^{t^2/2}\frac{\sqrt{\frac{2}{\pi}}e^{-t^2/2}}{t+\sqrt{t^2+8/\pi}}\Phi(t)+\frac{1}{\sqrt{2\pi}}(1-2\Phi(t))<0
$$
by simplification, we get
$$
2\Phi(t)-1>\frac{t}{\sqrt{t^2+8/\pi}}
$$
By applying the lemma again and further simplification, we get
$$
e^{t^2}-t^2>\frac{8}{\pi}
$$

This can easily be verified for $t=t_0$. Also LHS is increasing when $t\geq t_0$, therefore the lemma follows.

\end{proof}
%

\begin{lemma} \label{lem:onesolsuffices}
Let $f_1(x)$ and $f_2(x)$ be twice differentiable decreasing functions defined on $[0,\infty)$ satisfying the following conditions
\begin{enumerate}
\item
$f_1(0)=f_2(0)$
\item
$\lim_{x\to \infty} f_1(x) = \lim_{x\to \infty} f_2(x)$
\item
$\lim_{x\to 0} \frac{f_1'(x)}{f_2'(x)}>1$
\item
$\frac{f_1'(x)}{f_2'(x)}=1$ has only one solution
\end{enumerate}
then
$$
f_1(x)\leq f_2(x), \ \ \ \ \forall x\geq 0
$$
\end{lemma}

\begin{proof}
For the sake of contradiction we assume there exists $x_0$ such that $f_1(x_0)>f_2(x_0)$. By the mean value theorem, there exists $x_1<x_0$ such that $f_1'(x_1)>f_2'(x_1)$, which means $\frac{f_1'(x_1)}{f_2'(x_1)}< 1$ (since both $f_1'$ and $f_2'$ are negative). By the fourth assumption, for any $x>x_0>x_1$, $f_1'(x)>f_2'(x)$, therefore $f_1(x)-f_2(x)\geq f_1(x_0)-f_2(x_0)>0$, contradicting the second assumption.
\end{proof}

Now we show the key lemma in this section.

\begin{lemma} \label{lem:brvshs}
Let $u=\mu I+w_1$ and $v=\mu I+w_2$ be two unit vectors with the same
projection on the direction of $I$. Also we assume that $\langle
\bar{w}_1,\bar{w}_2 \rangle=1-\rho\geq 0$, where $\bar{w}_1$ and
$\bar{w}_2$ are the normalized vectors of $w_1$ and $w_2$. Then the
probability that these two vectors are separated by a random
hyperplane is at least the probability that these two vectors are cut by \pref{alg:rounding}.
\end{lemma}
\begin{proof}
First notice that since $u$ and $v$ have the same bias $\mu$, they
will be assigned the same threshold $t=\Phi^{-1}(2\mu-1)$ in \pref{alg:rounding}.

Henceforth, we fix  $\langle \bar{w_1},\bar{w_2} \rangle=1-\rho\geq 0$, and express the probabilities as a function of $\mu$ and  $t$. We stress that $\mu$ and $t$ are fully dependent on each other, therefore the functions are only single variable functions. We use both $\mu$ and $t$ (and other notations that are about to be introduced) in the expression only for simplicity.

Let $\eps=(1-\mu^2)\rho$, which characterizes $ \langle u,v\rangle$ as a function of $\mu$, \ie
$$
\iprod{u,v}=\iprod{\mu I+\sqrt{1-\mu^2}\bar{w}_1),(\mu
I+\sqrt{1-\mu^2}\bar{w}_2)} = 1-\eps
$$
Let $H(t)$ be the probability of the two vectors being separated by a
random hyperplane. It is well-known that \cite{GoemansW95}
$$
H(t)=\arccos(u\cdot v)/\pi=\arccos(1-\epsilon)/\pi
$$
For \pref{alg:rounding}, notice that $\bar{w}_1\cdot g$ and $\bar{w}_2\cdot g$ are two jointly distributed standard Gaussian variables with covariance matrix  $\Sigma=
\begin{pmatrix}
1\ \ \ \ \ 1-\rho \\ 1-\rho\ \ \ \ \ 1\\
\end{pmatrix}
$.
Thus the probability of $u$ and $v$ being separated by \pref{alg:rounding} is
$$
B(t)=2\int_{-\infty}^{t}\int_{t}^{\infty}\frac{1}{2\pi|\Sigma|^{1/2}}e^{-(x_1\ x_2)\Sigma^{-1}(x_1\ x_2)^{T}}dx_1dx_2
$$
It's easy to see that when $\mu=t=0$, these two rounding schemes are equivalent, thus $B(0)=H(0)$. Also $\lim_{t\to\infty} B(t)=\lim_{t\to\infty} H(t)=0$.
The derivatives of $H(t)$ and $B(t)$ are as follows:
$$
H'(t)=-\frac{2\sqrt{2}\rho}{\pi^{3/2}\sqrt{2\eps-\eps^2}}\tilde{\Phi}(t)e^{-t^2/2}
$$
and
$$
B'(t)=-\sqrt{\frac{2}{\pi}}\tilde{\Phi}(at)e^{-t^2/2}
$$

where $a=\frac{\rho}{\sqrt{2\rho-\rho^2}}\leq 1$ when $\rho\leq 1$, and $\tilde{\Phi}(t)$ is defined as
$$
\tilde{\Phi}(t)=\Phi(t)-\Phi(-t)
$$

Let $f(t)=\frac{B'(t)}{H'(t)}$. Notice that $f(0)=\pi/2\geq 1$, thus by
\pref{lem:onesolsuffices}, we only have to show that $f(t)=1$ has only one solution. Moreover, it suffices to show that $f'(t)<0$ when $f(t)\leq 1$.

Notice that when $f(t)\leq 1$, we have
\begin{align*}
\frac{\sqrt{2\epsilon-\epsilon^2}}{\sqrt{2\rho-\rho^2}}\frac{\tilde{\Phi}(at)}{a\tilde{\Phi}(t)}& \leq \frac{2}{\pi} \\
\Rightarrow \frac{2\epsilon-\epsilon^2}{2\rho-\rho^2}& \leq \frac{4}{\pi^2} \ \ \ \ \ \ \text{(By convexity of $\tilde{\Phi}$, $\frac{\tilde{\Phi}(at)}{a\tilde{\Phi}(t)}\geq 1$ when $a\leq 1$)}\\
\Rightarrow \frac{\epsilon}{\rho}\frac{2-\epsilon}{2-\rho}& \leq \frac{4}{\pi^2} \\
\Rightarrow  (1-\mu^2)\frac{2-\rho}{2-\epsilon}& \leq \frac{4}{\pi^2}
\ \ \ \ \ \ \ \left(\frac{\epsilon}{\rho}=1-\mu^2\right) \\
\Rightarrow \mu \geq \sqrt{1-4/\pi^2}&=\mu_0 \ \ \ \ \ \ \ \
\left(\frac{2-\rho}{2-\eps}\leq 1\right) \\
\Rightarrow t&\geq t_0
\end{align*}
By calculation, one can show that
$$
f'(t)=\frac{\sqrt{2/\pi}e^{-t^2/2}\sqrt{2\epsilon-\epsilon^2}}{\tilde{\Phi}(t)}\left(\frac{1-\epsilon}{2\epsilon-\epsilon^2}(-2\mu\rho)\tilde{\Phi}(at)+e^{(1-a^2)t^2/2}a-\frac{\tilde{\Phi}(at)}{\tilde{\Phi}(t)}\right)
$$

Now we show $f'(t)<0$ when $t\geq t_0$. In order to show this, one only needs to show that
$$
\frac{1-\epsilon}{2\epsilon-\epsilon^2}(2\mu\rho)\tilde{\Phi}(at)+\frac{\tilde{\Phi}(at)}{\tilde{\Phi}(t)}>e^{(1-a^2)t^2/2}a
$$
By substituting $\epsilon=(1-\mu^2)\rho$ and simplification, we get
$$
\frac{\tilde{\Phi}(at)}{a\tilde{\Phi}(t)}\frac{1}{1-\mu^2}\left(\frac{1-\epsilon}{2-\epsilon}2\mu^2+1-\mu^2\right)\geq e^{(1-a^2)t^2/2}
$$
Since $\frac{\tilde{\Phi}(at)}{a\tilde{\Phi}(t)}\geq 1$ when $a\leq 1$ and $e^{(1-a^2)t^2/2}\leq e^{t^2/2}$, it suffices to show
$$
\left(2\mu^2\frac{1-\epsilon}{2-\epsilon}+1-\mu^2\right)\geq e^{t^2/2}(1-\mu^2)
$$
holds when $t\geq t_0$.

By \pref{lem:monotonicity}, we know that RHS is decreasing when $t\geq t_0$. Now we show LHS is increasing when $\mu\geq \mu_0$. It can be shown that the derivative of LHS is
$$
2\mu\rho(1-\mu^2)\mu^2-(2\mu-4\mu^3)(2-\epsilon)\geq -\mu(2-4\mu^2)(2-\epsilon)\geq 0
$$
when $\mu\geq \mu_0$.

Now we only have to verify the inequality when $t=t_0$, and that can be done numerically. The calculation shows that  $\text{LHS}(t_0)\approx 0.8489$ while $\text{RHS}(t_0)\approx 0.836$.

\end{proof}
Finally, we show \pref{lem:rooteps}.
\begin{lemma} (Restatement of \pref{lem:rooteps})
Let $u=\mu_1 I+ w_1$,$v=\mu_2 I + w_2$ be two unit vectors satisfying $\|u-v\|^2/4\leq \eps$, then the probability of them being separated by \pref{alg:rounding} is at most $O(\sqrt{\eps})$.
\end{lemma}
\begin{proof} (Proof of \pref{lem:rooteps})

First we prove the case when $\mu_1=\mu_2=\mu$. Notice that when $\langle w_1, w_2\rangle>0$, the lemma follows from \pref{lem:brvshs} and the fact that Goemans-Williamson algorithm will separate $u$ and $v$ with probability $O(\sqrt{\eps})$\cite{GoemansW95}.

If $\langle w_1,w_2\rangle<0$, then $\|u-v\|^2/4=\|w_1-w_2\|^2/4\geq (\|w_1\|^2+\|w_2\|^2)/4=(1-\mu^2)/2$. Hence $|\mu|\geq 1-O(\sqrt{\eps})$. By union bound, the probability of the algorithm separating $u$ and $v$ is at most $O(\sqrt{\eps})$.

Now we consider the case when $\mu_1\neq \mu_2$, w.l.o.g. we may assume $|\mu_1|>|\mu_2|$.  We construct an auxiliary vector $v'$ as follow: $v'=\mu_1 I+\sqrt{1-\mu_1^2}\bar{w}_2$. It's easy to see that $\|u-v'\|\leq \|u-v\|$.  Let $F$ denote the rounding function, we analyze the probability of $u$ and $v$ being separated as follows:
\begin{align*}
&\ \Pr(F(u)\neq F(v)) \\
&=\Pr(F(u) \neq F(v'), F(v')=F(v))+\Pr(F(u)=F(v'),F(v')\neq F(v)) \\
&\leq \Pr(F(u)\neq F(v'))+\Pr(F(v')\neq F(v))
\end{align*}
Since $\|u-v'\|\leq \|u-v\|$ and $\langle u,I\rangle = \langle v',I\rangle=\mu_1$, by the first part of the proof
$
\Pr(F(u)\neq F(v'))\leq O(\sqrt{\eps})
$.
Also,
$$
\Pr(F(v')\neq F(v))\leq |\mu_1-\mu_2|/2 \leq \|u-v\|/2 \leq
O(\sqrt{\eps}) \mper
$$
Therefore the lemma follows.
\end{proof}

\section{Mutual Information, Statistical Distance and Independence}
\label{app:informationtheory}

Intuitively, when two random variables have low mutual information, they should be close to being independent. In this section we formalize this intuition by giving an explicit bound on the statistical distance between the joint distribution and the independent distribution. We stress that all the results here are sufficient for our use in this work, but we believe the parameters could be further optimized.

We start by defining a few notions that measures the correlation of two random variables.
\begin{definition}
Let $\Omega$ be a finite sample space, $P$ and $Q$ be two probability distributions on $\Omega$. The \emph{square Hellinger distance} of $P$ and $Q$ is defined as
$$
H^2(P,Q)=\frac{1}{2}\sum_{x\in \Omega}(\sqrt{P(x)}-\sqrt{Q(x)})^2
$$
\end{definition}
\begin{definition}
Let $\Omega$ be a finite sample space, $P$ and $Q$ be two probability distributions on $\Omega$. The \emph{Kullback-Leibler divergence} of $P$ and $Q$ is defined as
$$
D_{KL}(P\|Q)=\sum_{x\in\Omega}P(x)\log\frac{P(x)}{Q(x)}
$$
\end{definition}
Now we give a few facts regarding mutual information, Hellinger distance and Kullback-Leibler divergence without proving them.
\begin{fact}
Let $X$ and $Y$ be two jointly distributed random variables taking value in $[q]$, then
$$
I(X;Y)=D_{KL}(p(x,y)\|p(x)\times p(y)).
$$
where $p(x,y)$ is the joint distribution of $X$ and $Y$ on $[q]^2$ and $p(x)\times p(y)$ is the product distribution of the marginal distributions of $X$ and $Y$.
\end{fact}

\begin{fact}
Let $\Omega$ be a finite sample space, $P$ and $Q$ be two probability distribution on $\Omega$, then
$$
D_{KL}(Q\|P)\geq \frac{2}{\ln 2}H^2(P,Q)
$$
\end{fact}
Combining the facts mentioned above, we get the following relation
between mutual information and statistical distance.
\begin{fact} (Restatement of \pref{fact:statdist})
Let $X$ and $Y$ be two jointly distributed random variables on $[q]$
then,
$$ I(X;Y)\geq \frac{1}{2\ln2}\sum_{i,j\in
[q]}(\mathbb{P}(X=i,Y=j)-\mathbb{P}(X=i)\mathbb{P}(Y=j))^2 \mcom $$
in particular for all $i,j\in [q]$
$$
|\mathbb{P}(X=i,Y=j)-\mathbb{P}(X=i)\mathbb{P}(Y=j)|\leq \sqrt{2 I(X;Y)}
$$
As a consequence, if $X$ and $Y$ are two random variables defined on $\{-1,1\}$,
$
\mbox{Cov}(X,Y)\leq O(\sqrt{I(X;Y)})
$
\end{fact}
%
%
\begin{proof}
\begin{align*}
I(X;Y)&=D_{KL}(p(x,y)\|p(x)\times p(y)) \\
      &\geq \frac{2}{\ln 2}H^2(p(x,y),p(x)\times p(y)) \\
      & =\frac{2}{\ln 2}\sum_{i,j\in
      [q]}\left(\sqrt{\mathbb{P}(X=i,Y=j)}-\sqrt{\mathbb{P}(X=i)\mathbb{P}(Y=j)}\right)^2 \\
      & =\frac{2}{\ln 2}\sum_{i,j\in
      [q]}\left(\frac{\mathbb{P}(X=i,Y=j)-\mathbb{P}(X=i)\mathbb{P}(Y=j)}{\sqrt{\mathbb{P}(X=i,Y=j)}+\sqrt{\mathbb{P}(X=i)\mathbb{P}(Y=j)}}\right)^2 \\
      & \geq \frac{1}{2\ln2}\sum_{i,j\in [q]}(\mathbb{P}(X=i,Y=j)-\mathbb{P}(X=i)\mathbb{P}(Y=j))^2
\end{align*}
Upper bounding $ln2$ by 1, finishes the proof.
\end{proof}

\fi

\end{document}

%% file: macros.tex
\usepackage{etex}

\usepackage{xspace,enumerate}

\usepackage[dvipsnames]{xcolor}

\usepackage[T1]{fontenc}
\usepackage[full]{textcomp}

\usepackage[american]{babel}

\usepackage{mathtools}

\usepackage{bm}

\usepackage{amsthm}

\newtheorem{theorem}{Theorem}[section]
\newtheorem*{theorem*}{Theorem}

\newtheorem*{claim*}{Claim}

\newtheorem{proposition}[theorem]{Proposition}
\newtheorem*{proposition*}{Proposition}
\newtheorem{lemma}[theorem]{Lemma}
\newtheorem*{lemma*}{Lemma}
\newtheorem{corollary}[theorem]{Corollary}

\newtheorem*{conjecture*}{Conjecture}

\newtheorem{fact}[theorem]{Fact}
\newtheorem*{fact*}{Fact}

\newtheorem*{hypothesis*}{Hypothesis}

\theoremstyle{definition}
\newtheorem{definition}[theorem]{Definition}

\newtheorem{algorithm}[theorem]{Algorithm}

\newtheorem{remark}[theorem]{Remark}

\usepackage[letterpaper,
top=1.3in,
bottom=1.3in,
left=1.6in,
right=1.6in]{geometry}

\usepackage[varg]{txfonts} 
\renewcommand{\mathbb}{\varmathbb}

\ifnum\showkeys=1
\usepackage[color]{showkeys}
\fi

\ifnum\showcolorlinks=1
\usepackage[
pagebackref,
letterpaper=true,
colorlinks=true,
urlcolor=blue,
linkcolor=blue,
citecolor=OliveGreen,
]{hyperref}
\fi

\ifnum\showcolorlinks=0
\usepackage[
pagebackref,
letterpaper=true,
colorlinks=false,
pdfborder={0 0 0}
]{hyperref}
\fi

\usepackage{prettyref}

\newcommand{\savehyperref}[2]{\texorpdfstring{\hyperref[#1]{#2}}{#2}}

\newrefformat{eq}{\savehyperref{#1}{\textup{(\ref*{#1})}}}
\newrefformat{lem}{\savehyperref{#1}{Lemma~\ref*{#1}}}
\newrefformat{def}{\savehyperref{#1}{Definition~\ref*{#1}}}
\newrefformat{thm}{\savehyperref{#1}{Theorem~\ref*{#1}}}
\newrefformat{cor}{\savehyperref{#1}{Corollary~\ref*{#1}}}
\newrefformat{cha}{\savehyperref{#1}{Chapter~\ref*{#1}}}
\newrefformat{sec}{\savehyperref{#1}{Section~\ref*{#1}}}
\newrefformat{app}{\savehyperref{#1}{Appendix~\ref*{#1}}}
\newrefformat{tab}{\savehyperref{#1}{Table~\ref*{#1}}}
\newrefformat{fig}{\savehyperref{#1}{Figure~\ref*{#1}}}
\newrefformat{hyp}{\savehyperref{#1}{Hypothesis~\ref*{#1}}}
\newrefformat{alg}{\savehyperref{#1}{Algorithm~\ref*{#1}}}
\newrefformat{rem}{\savehyperref{#1}{Remark~\ref*{#1}}}
\newrefformat{item}{\savehyperref{#1}{Item~\ref*{#1}}}
\newrefformat{step}{\savehyperref{#1}{step~\ref*{#1}}}
\newrefformat{conj}{\savehyperref{#1}{Conjecture~\ref*{#1}}}
\newrefformat{fact}{\savehyperref{#1}{Fact~\ref*{#1}}}
\newrefformat{prop}{\savehyperref{#1}{Proposition~\ref*{#1}}}
\newrefformat{prob}{\savehyperref{#1}{Problem~\ref*{#1}}}
\newrefformat{claim}{\savehyperref{#1}{Claim~\ref*{#1}}}
\newrefformat{relax}{\savehyperref{#1}{Relaxation~\ref*{#1}}}
\newrefformat{red}{\savehyperref{#1}{Reduction~\ref*{#1}}}
\newrefformat{part}{\savehyperref{#1}{Part~\ref*{#1}}}

\newcommand{\Sref}[1]{\hyperref[#1]{\S\ref*{#1}}}

\usepackage{nicefrac}

\ifnum\usemicrotype=1
\usepackage{microtype}
\fi

\ifnum\showauthornotes=1
\newcommand{\Authornote}[2]{{\sffamily\small\color{red}{[#1: #2]}}}
\newcommand{\Authorcomment}[2]{{\sffamily\small\color{gray}{[#1: #2]}}}
\newcommand{\Authorstartcomment}[1]{\sffamily\small\color{gray}[#1: }

\newcommand{\Authorfnote}[2]{\footnote{\color{red}{#1: #2}}}
\newcommand{\Authorfixme}[1]{\Authornote{#1}{\textbf{??}}}
\newcommand{\Authormarginmark}[1]{\marginpar{\textcolor{red}{\fbox{\Large #1:!}}}}
\else
\newcommand{\Authornote}[2]{}
\newcommand{\Authorcomment}[2]{}
\newcommand{\Authorstartcomment}[1]{}

\newcommand{\Authorfnote}[2]{}
\newcommand{\Authorfixme}[1]{}
\newcommand{\Authormarginmark}[1]{}
\fi

\ifnum\showfixme=0

\fi

\usepackage{boxedminipage}

\newenvironment{mybox}
{\center \noindent\begin{boxedminipage}{1.0\linewidth}}
{\end{boxedminipage}
\noindent
}

\newcommand{\norm}[1]{\lVert#1\rVert}

\newcommand{\iprod}[1]{\langle#1\rangle}

\newcommand{\Esymb}{\mathbb{E}}
\newcommand{\Psymb}{\mathbb{P}}
\newcommand{\Vsymb}{\mathbb{V}}

\DeclareMathOperator*{\E}{\Esymb}
\DeclareMathOperator*{\Var}{\Vsymb}
\DeclareMathOperator*{\ProbOp}{\Psymb}

\renewcommand{\Pr}{\ProbOp}

\newcommand{\textparen}[1]{\text{(#1)}}

\ifx\because\undefined
\newcommand{\because}[1]{\textparen{because #1}}
\else
\renewcommand{\because}[1]{\textparen{because #1}}
\fi

\newcommand{\sbits}{\{\pm1\}}

\newcommand{\defeq}{\stackrel{\mathrm{def}}=}

\renewcommand{\vec}[1]{{\bm{#1}}}

\newcommand{\mper}{\,.}
\newcommand{\mcom}{\,,}

\newcommand\bdot\bullet

\renewcommand{\th}{\textsuperscript{th}\xspace}

\DeclareMathOperator{\Inf}{Inf}

\DeclareMathOperator{\val}{val}

\DeclareMathOperator{\OPT}{OPT}
\DeclareMathOperator{\opt}{opt}

\DeclareMathOperator{\poly}{poly}

\newcommand{\ie}{i.e.,\xspace}
\newcommand{\eg}{e.g.,\xspace}

\newcommand{\etal}{et al.\xspace}

\newcommand{\R}{\mathbb R}

\newcommand{\problemmacro}[1]{\texorpdfstring{\textsc{#1}}{#1}\xspace}

\newcommand{\uniquegames}{\problemmacro{Unique Games}}
\newcommand{\maxcut}{\problemmacro{Max Cut}}

\newcommand{\balancedseparator}{\problemmacro{Balanced Separator}}
\newcommand{\maxtwosat}{\problemmacro{Max $2$-Sat}}
\newcommand{\maxthreesat}{\problemmacro{Max $3$-Sat}}

\newcommand{\smallsetexpansion}{\problemmacro{Small-Set Expansion}}
\newcommand{\minbisection}{\problemmacro{Min Bisection}}

\newcommand{\maxbisection}{\problemmacro{Max Bisection}}

\newcommand{\cF}{\mathcal F}

\renewcommand{\leq}{\leqslant}

\renewcommand{\geq}{\geqslant}

\ifnum\showdraftbox=1
\newcommand{\draftbox}{\begin{center}
  \fbox{%
    \begin{minipage}{2in}%
      \begin{center}%
          \Large\textsc{Working Draft}\\%
        Please do not distribute%
      \end{center}%
    \end{minipage}%
  }%
\end{center}
\vspace{0.2cm}}
\else
\newcommand{\draftbox}{}
\fi

\let\epsilon=\varepsilon

\numberwithin{equation}{section}

\newcommand{\MYstore}[2]{%
  \global\expandafter \def \csname MYMEMORY #1 \endcsname{#2}%
}

\newcommand{\MYload}[1]{%
  \csname MYMEMORY #1 \endcsname%
}

\newcommand{\MYnewlabel}[1]{%
  \newcommand\MYcurrentlabel{#1}%
  \MYoldlabel{#1}%
}

\newcommand{\MYdummylabel}[1]{}

\newcommand{\torestate}[1]{%
  \let\MYoldlabel\label%
  \let\label\MYnewlabel%
  #1%
  \MYstore{\MYcurrentlabel}{#1}%
  \let\label\MYoldlabel%
}

\newcommand{\restatetheorem}[1]{%
  \let\MYoldlabel\label
  \let\label\MYdummylabel
  \begin{theorem*}[Restatement of \prettyref{#1}]
    \MYload{#1}
  \end{theorem*}
  \let\label\MYoldlabel
}

\newcommand{\restatelemma}[1]{%
  \let\MYoldlabel\label
  \let\label\MYdummylabel
  \begin{lemma*}[Restatement of \prettyref{#1}]
    \MYload{#1}
  \end{lemma*}
  \let\label\MYoldlabel
}

\newcommand{\restateprop}[1]{%
  \let\MYoldlabel\label
  \let\label\MYdummylabel
  \begin{proposition*}[Restatement of \prettyref{#1}]
    \MYload{#1}
  \end{proposition*}
  \let\label\MYoldlabel
}

\newcommand{\restatefact}[1]{%
  \let\MYoldlabel\label
  \let\label\MYdummylabel
  \begin{fact*}[Restatement of \prettyref{#1}]
    \MYload{#1}
  \end{fact*}
  \let\label\MYoldlabel
}

\newcommand{\restate}[1]{%
  \let\MYoldlabel\label
  \let\label\MYdummylabel
  \MYload{#1}
  \let\label\MYoldlabel
}

\newcommand{\addreferencesection}{
  \phantomsection
  \addcontentsline{toc}{section}{References}
}

\newcommand{\sse}{\subseteq}

\newcommand{\eps}{\epsilon}

\let\origparagraph\paragraph
\renewcommand{\paragraph}[1]{\origparagraph{#1.}}

\allowdisplaybreaks

\newcommand{\cclassmacro}[1]{\texorpdfstring{\textbf{#1}}{#1}\xspace}

\newcommand{\NP}{\cclassmacro{NP}}

\newcommand{\dict} {\textsf{DICT}}

\newcommand{\mrv}[1]{\bm{\lowercase{#1}}}

\newcommand{\erv}[1]{\mathcal{#1}}

\newcommand{\orf}[1]{\mathcal{\uppercase{#1}}}

\newcommand{\mrf}[1]{\bm{\mathcal{\uppercase{#1}}}}

\newcommand{\opl}[1]{{\uppercase{#1}}}

\newcommand{\msf}[1]{\mathsf{#1}}
\newcommand{\mcl}[1]{\mathcal{#1}}

\newcommand{\round}{\msf{Round}}
\newcommand{\struncate}{f_{[-1,1]}}
\newcommand{\T}{\mathrm{T}}

%% file: max-bisection.bbl
\newcommand{\etalchar}[1]{$^{#1}$}
\providecommand{\bysame}{\leavevmode\hbox to3em{\hrulefill}\thinspace}
\providecommand{\MR}{\relax\ifhmode\unskip\space\fi MR }
\providecommand{\MRhref}[2]{%
  \href{http://www.ams.org/mathscinet-getitem?mr=#1}{#2}
}
\providecommand{\href}[2]{#2}
\begin{thebibliography}{KKMO07}

\bibitem[ARV04]{AroraRV04}
Sanjeev Arora, Satish Rao, and Umesh Vazirani, \emph{Expander flows, geometric
  embeddings and graph partitioning}, Proceedings of the thirty-sixth annual
  {ACM} Symposium on Theory of Computing ({STOC}-04) (New York), ACM Press,
  June ~13--15 2004, pp.~222--231.

\bibitem[Aus07]{Austrin07a}
Per Austrin, \emph{Balanced max 2-sat might not be the hardest}, Proceedings of
  the 39th Annual {ACM} Symposium on Theory of Computing, San Diego,
  California, {USA}, June 11-13, 2007 (David~S. Johnson and Uriel Feige, eds.),
  ACM, 2007, pp.~189--197.

\bibitem[BRS11]{BarakRS11}
Boaz Barak, Prasad Raghavendra, and David Steurer, \emph{Rounding semidefinite
  programming hierarchies via global correlation}, FOCS (to appear), 2011.

\bibitem[CMM06]{CharikarMM06e}
Moses Charikar, Yury Makarychev, and Konstantin Makarychev, \emph{Near-optimal
  algorithms for unique games}, STOC: ACM Symposium on Theory of Computing
  (STOC), 2006.

\bibitem[CMM07]{CharikarMM07a}
Moses Charikar, Konstantin Makarychev, and Yury Makarychev, \emph{Near-optimal
  algorithms for maximum constraint satisfaction problems}, Proceedings of the
  Eighteenth Annual {ACM}-{SIAM} Symposium on Discrete Algorithms, {SODA} 2007,
  New Orleans, Louisiana, {USA}, January 7-9, 2007 (Nikhil Bansal, Kirk Pruhs,
  and Clifford Stein, eds.), SIAM, 2007, pp.~62--68.

\bibitem[FJ97]{FriezeJ97}
Alan~M. Frieze and Mark Jerrum, \emph{Improved approximation algorithms for max
  k-cut and max bisection}, Algorithmica \textbf{18} (1997), no.~1, 67--81.

\bibitem[FL06]{FeigeL06}
Uriel Feige and Michael Langberg, \emph{The {RPR}$^{\mbox{2}}$ rounding
  technique for semidefinite programs}, J. Algorithms \textbf{60} (2006),
  no.~1, 1--23.

\bibitem[GMR{\etalchar{+}}11]{GuruswamiMRSZ11}
Venkatesan Guruswami, Yury Makarychev, Prasad Raghavendra, David Steurer, and
  Yuan Zhou, \emph{Finding almost perfect graph bisections}, Innovations in
  Computer Science, Tsinghua University Press, 2011, pp.~321--337.

\bibitem[GS11]{GuruswamiS11}
Venkatesan Guruswami and Ali~Kemal Sinop, \emph{Lasserre hierarchy, higher
  eigenvalues, and approximation schemes for quadratic integer programming with
  psd objectives}, FOCS (to appear), 2011.

\bibitem[GW95]{GoemansW95}
Michel~X. Goemans and David~P. Williamson, \emph{Improved approximation
  algorithms for maximum cut and satisfiability problems using semidefinite
  programming}, Journal of the ACM \textbf{42} (1995), no.~6, 1115--1145.

\bibitem[H\.01]{Hastad01}
Johann H\.astad, \emph{Some optimal inapproximability results}, Journal of the
  ACM \textbf{48} (2001), no.~4, 798--859.

\bibitem[HK04]{HolmerinK04}
Jonas Holmerin and Subhash Khot, \emph{A new {PCP} outer verifier with
  applications to homogeneous linear equations and max-bisection}, Proceedings
  of the 36th Annual ACM Symposium on Theory of Computing, 2004, pp.~11--20.

\bibitem[HZ02]{HalperinZ02}
Eran Halperin and Uri Zwick, \emph{A unified framework for obtaining improved
  approximation algorithms for maximum graph bisection problems}, Random
  Struct. Algorithms \textbf{20} (2002), no.~3, 382--402.

\bibitem[IM09]{IsakssonM09}
Marcus Isaksson and Elchanan Mossel, \emph{Maximally stable gaussian partitions
  with discrete applications}, arXiv:0903.3362. (2009).

\bibitem[KKMO07]{KhotKMO07}
Subhash Khot, Guy Kindler, Elchanan Mossel, and Ryan O'Donnell, \emph{Optimal
  inapproximability results for max-cut and other 2-variable csps?}, SIAM J.
  Comput. \textbf{37} (2007), no.~1, 319--357.

\bibitem[Kom55]{Komatsu55}
Y.~Komatsu, \emph{Elementary inequalities for mills' ratio}, Rep. Statist.
  Appl. Res. Un. Japan. Sci. Engrs (1955).

\bibitem[KS09]{KhotS09}
Subhash Khot and Rishi Saket, \emph{Sdp integrality gaps with local
  $ell_1$-embeddability}, FOCS, 2009, pp.~565--574.

\bibitem[KZ97]{KarloffZ97}
Howard Karloff and Uri Zwick, \emph{A 7/8-approximation algorithm for {MAX}
  3{SAT}?}, Proceedings of the 38th Annual IEEE Symposium on Foundations of
  Computer Science, FOCS'97 (Miami Beach, Florida, October 20-22, 1997) (Los
  Alamitos-Washington-Brussels-Tokyo), IEEE Computer Society, IEEE Computer
  Society Press, 1997, pp.~406--415.

\bibitem[Las01]{Lasserre01}
J.~B. Lasserre, \emph{An explicit exact {SDP} relaxation for nonlinear 0-1
  programs}, IPCO 2001 (K.~Aardel and A.~M.~H. Gerards, eds.), Lecture Notes in
  Computer Science, vol. 2081, Springer, Berlin, 2001, pp.~293--303.

\bibitem[Rag08]{Raghavendra08}
Prasad Raghavendra, \emph{Optimal algorithms and inapproximability results for
  every csp?}, In STOC '08: Proceedings of the 40th ACM Symposium on Theory of
  Computing (2008), 245--254.

\bibitem[RS09a]{RaghavendraS09b}
Prasad Raghavendra and David Steurer, \emph{How to round any csp}, FOCS, 2009,
  pp.~586--594.

\bibitem[RS09b]{RaghavendraS09c}
\bysame, \emph{Integrality gaps for strong sdp relaxations of unique games},
  FOCS, 2009, pp.~575--585.

\bibitem[RS10]{RaghavendraS10}
\bysame, \emph{Graph expansion and the unique games conjecture}, STOC, 2010,
  pp.~755--764.

\bibitem[TSSW00]{TrevisanSSW00}
Luca Trevisan, Gregory~B. Sorkin, Madhu Sudan, and David~P. Williamson,
  \emph{Gadgets, approximation, and linear programming}, SIAM J. Comput
  \textbf{29} (2000), no.~6, 2074--2097.

\bibitem[Tul09]{Tulsiani09}
Madhur Tulsiani, \emph{{CSP} gaps and reductions in the {L}asserre hierarchy},
  STOC, 2009, To apperar.

\bibitem[Ye01]{Ye01}
Yinyu Ye, \emph{A .699-approximation algorithm for {M}ax-{B}isection},
  Mathematical Programming \textbf{90} (2001), 101--111.

\end{thebibliography}
